\documentclass{article}
\usepackage{a4wide}

\usepackage{graphicx} % Required for inserting images
\usepackage{amsmath,amsfonts,amsthm,amssymb}
\usepackage{xspace}

\newcommand{\R}{\mathbb{R}}
\newcommand{\A}{\ensuremath{\mathbb{A}}\xspace}
\newcommand{\B}{\ensuremath{\mathbb{B}}\xspace}
\newcommand{\X}{\ensuremath{\mathbb{X}}\xspace}

\newcommand{\re}{\mathrm{e}}
\newcommand{\rd}{\mathrm{d}}
\newcommand{\Ncal}{\mathcal{N}}
\newcommand{\xs}{x^*}%{\ensuremath{x^*}}

\newcommand{\xmin}{\ensuremath{x_{\mathrm{min}}}}
\newcommand{\order}[1]{\mathcal{O}\left(#1\right)}
\newcommand{\bydef}{:=}

\newcommand{\ba}{\mathbf{a}}
\newcommand{\bb}{\mathbf{b}}
\newcommand{\smallo}[1]{\ensuremath{\mathrm{o}\!\left(#1\right)}}
\newcommand{\N}{\mathbb{N}}
\newcommand{\T}{\mathbb{T}}

\newcommand{\II}{\mathcal{I}}
 
\newcommand{\Vgap}[1]{\mathfrak{G}(#1)}
\newcommand{\Msep}[1]{\mathfrak{M}(#1)}

\newtheorem{lemma}{Lemma}
\newtheorem{theorem}{Theorem}
\newtheorem{corollary}{Corollary}
\newtheorem{remark}{Remark}
\newtheorem{definition}{Definition}

\usepackage{authblk}

\title{Dynamical compatibility for finite and infinite population models used in genetics}

\author[1]{Fabio A. C. C. Chalub\footnote{Corresponding author: facc@fct.unl.pt}}

\author[2]{Max O. Souza\footnote{maxsouza@id.uff.br}}

\affil[1]{Center for Mathematics and Applications (NOVA Math) and Department of Mathematics, NOVA School of Science and Technology, Universidade NOVA de Lisboa, Quinta da Torre, 2829-516, Caparica, Portugal.}

\affil[2]{Instituto de Matemática e Estatística, Universidade Federal Fluminense, R. Prof. Marcos Waldemar de Freitas Reis, s/n, Bloco H, Niterói, RJ 24120-201, Brasil}

\date{\today}

\begin{document}

\maketitle

\begin{abstract}
    Finite and infinite population models are frequently used in population dynamics. However, their interrelationship is rarely discussed. 
    
    \textbf{Relevance to life sciences:} In this work, we examine the limits of large populations of the Moran process (a finite-population birth-death process) and the replicator equation (an ordinary differential equation) as paradigmatic examples of finite and infinite population models, respectively, both of which are extensively used in population genetics. Except for certain degenerate cases, we completely characterize when these models exhibit similar long-term dynamics, i.e., when there is a one-to-one relation between the stable attractors of the replicator equations and the metastable states of the Moran process.  We will discuss our findings and the implications for an improved understanding of the famous Lewontin paradox.
    
    \textbf{Mathematical content:} To achieve this goal, we first show that the asymptotic expression for the fixation probability in the Moran process, when the population size is large and individual interaction is almost arbitrary (including cases modeled through $d$-player game theory), is a convex combination of the asymptotic approximations obtained in the constant fitness case or 2-player game theory. We discuss several examples and the inverse problem, i.e., how to derive a Moran process that is compatible with a given replicator dynamics. In particular, we prove that modeling a Moran process with an inner metastable state may require the use of $d$-player game theory with possibly large $d$ values, depending on the precise location of the inner equilibrium.  We finish by discussing possible connections between the newly introduced concepts and classical concepts of quasistationarity and metastability in the theory of stochastic processes.
\end{abstract}



\section{Introduction}

A real population can be so large that, even if we fully understood how individuals interact, it would be impractical to simulate its evolution numerically. In this scenario, there are two viable approaches: using infinite population models or conducting simulations of much smaller populations and extrapolating the results to realistic population sizes. However, the relationship between these two fundamentally different approaches remains unclear.

The main question we will study in this work is to what extent infinite population models serve as approximations for models with finite but large populations. More specifically, we will consider two well-known classes of models that describe the evolution of a population composed of two types of individuals. The first is the Moran process, a constant population birth-death process that allows some important exact formulas~\cite{Moran}. The second is the replicator equation, a first-order differential equation introduced in \cite{TaylorJonker_1978}.

In both cases, the resulting dynamics will strongly depend on the reproductive viability of both types of individuals in all possible configurations of the population, which we will call from now on \emph{population states}, or, simply, \emph{states}. We will henceforth assume one type as \emph{focal} (type \A), the other being called \B. In the Moran process, population size is constant, denoted by $N\in\N$, and the state of the population will be identified by $i\in\{0,\dots,N\}$, the number of type \A individuals in the population, or $x\in[0,1]$, its fraction. The latter will also be used to identify the state of the population in the replicator equation.  

The reproductive viability, or \emph{fitness}, is commonly calculated using game theory to describe individual interactions, being equal to the average game payoff. Although fitness can be defined for each population state without direct reference to individual interactions, the application of game theory provides valuable insights into how individuals interact with each other and the consequences of these interactions at the population level. 

For the finite population level, it is important to discuss how results for relatively small populations, i.e., composed by $N\approx 10^5-10^8$ individuals, and that allows direct numerical simulations, scale up to realistically large values of $N\approx 10^{12}$. More particularly, a topic of interest in the present work is how these parameters change as the population size $N$ increases, so that the infinite population limit is consistent. The relation between the parameters is called \emph{scaling laws}~\cite{ChalubSouza:TPB2009}. 

We will consider a function \( V \), referred to as the \emph{fitness potential}, which is defined for all possible population states, such that its derivative equals, minus the sign, the fitness differences among the types~\cite{ChalubSouza:JTB2018}. When \( V \) is a polynomial of degree \( d \), the individual interaction corresponds to the payoffs of a \( d' \)-player game, $d'\ge d$, with strict inequality in degenerated cases. We say that a $d$-player game is \emph{degenerated} if its dynamics is equivalent to a $d'$-player game, with $d'<d$; see~\cite{ChalubSouza:DGA2025} for an explanation of the case $d'=2$. For further details on game theory, see~\cite{Gintis}; the relationship between game theory and the replicator equation is discussed in~\cite{HofbauerSigmund}, while for the Moran process, cf.~\cite{Nowak:06}.

We will introduce a set of scaling assumptions that guarantee that the replicator equation is the formal limit of the Moran process. In the process of deriving a continuous approximation for the large-population Moran process, two important model parameters need to be constrained to the population size. 

The first is how the typical birth and death time of the Moran process, $\Delta t>0$, decreases as $N$ increases, while the second is the so-called \emph{weak selection principle}. Assuming, for a second, that $\Delta t$ is inversely proportional to $N$, then the renewal of the entire population will be size independent, and therefore the large population limit of the process will include, in a well-balanced way, both the natural selection and the genetic drift. In order to focus on natural selection --- as does the replicator equation --- it is necessary to increase the population size faster than in the balanced case discussed above, or, equivalently, decrease the time step slower than $N^{-1}$, as $N$ goes to infinity.

The second critical assumption, the \emph{weak selection principle}, says that when the time step in the Moran process is small --- meaning the population size is large --- the fitness difference between the types approaches one. 

For several aspects of the weak selection principle, cf.~\cite{Kimura1983,OhtaPNAS_2002,WildTraulsen_2007,WuTraulsen_2010} and for the dependence of the limiting process in the scalings, see~\cite{ ChalubSouza:TPB2009,Chalubsouza:JMB2014,MckaneWaxman_2007,TraulsenClaussenHauert_2005}.

With these assumptions on hand, we will develop the large population limit models of the Moran process and demonstrate that its relationship to the replicator equation is not as strong as one might expect. Specifically, we will show that some of the mixed stable states in the replicator equation, where both type fitnesses are equal, do not necessarily correspond to metastable states in the Moran process, even when the population size is large. We will be able to fully characterize which equilibrium points of the replicator equation are associated with metastability in the Moran process for large populations in terms of the potential \( V \). These specific equilibria will be referred to as \emph{essentially metastable states}, or EMS, generalizing to mixed cases in large populations the \(\mathrm{ESS}_N\), originally introduced for homogeneous populations in~\cite{Taylor_etal:BMB2004}.

Contrary to what one may expect from the use of the word \emph{metastability} in the previous paragraph, the focus of the present study is not the time-evolution of any of these processes, but the comparison of fixation probability of the focal time in the long run of the Moran process \emph{versus} the basin of attraction of the replicator dynamics. From the no-mutation assumption of the Moran process, \emph{like begets like}, it is not difficult to prove that, given enough time, the population will become homogeneous~\cite{ChalubSouza:JMB2017,Moran,Nowak:06}. The fixation probability, which depends on the initial condition, is the probability that the final state is composed only of individuals of the focal type. In this case, we say that the focal type \emph{fixated}; otherwise, we say that it became \emph{extinct}. When the dependence of the fixation probability is weakly dependent on the initial condition, for a given region, there was a loss of memory of the initial condition before fixation or extinction (in short, \emph{absorption}), indicating the existence of a metastable attractor in which both types are present. Only after this first time scale, dominated by the replicator dynamics, did the stochastic process move to either absorbed states.  Similarly, in the replicator dynamics, the basin of attraction of any possible stable equilibrium $X_*\in[0,1]$ is the set of initial conditions that guarantee that in the long run the state of the population will converge to $X_*$.

On the contrary, let a given equilibrium $X_*$ of the replicator equation be unstable, i.e., small perturbations of the initial condition lead to very different final states. In the Moran process, this corresponds to discontinuities (sharp variations) in the fixation probability, as closely related initial conditions lead to very different final states.

The main question that we will deal with is: what are the conditions on the parameters of the finite population Moran process, including fitness and scaling laws, such that when we scale up the population size to infinity i) fixation and extinction in the finite population implies attraction by the boundaries of the domain, $X_*=1$ and $X_*=0$ in the replicator equations, respectively, ii) plateaux in the fixation probability suggest the existence of interior stable equilibria $X_*\in(0,1)$ of the replicator dynamics, and iii) discontinuities of the fixation probability relate to interior unstable equilibria of the replicator dynamics?

With that problem in mind, we will start to show that the asymptotic behavior in the population size $N$ of the fixation probability for very general interactions in the Moran process, including the ones given by $d$-player game theory, can be expressed as a convex combination of two different functions: the first function arises from the fixation probability in cases with constant fitness differences (one-player games), while the second comes from coordination games in two-player game theory. 

After providing an almost complete answer to our main question, in which only degenerated cases will not be worked out in detail, we will discuss the inverse problem, i.e., we will explicitly construct fitness functions that, when implemented in the Moran process, make its dynamics similar to those in a given replicator dynamics. However, despite the simplicity in the mathematical problem, we show that unexpected results in designing finite population two-player games arise in games as simple as the hawk-dove game when the inner equilibrium $\xs\in(0,1)$ is such that $\xs\ne1/2$. In particular, looking for hawk-dove games with non-symmetric inner equilibrium in the finite population case, we cannot model individual interaction using two-player game theory, as is customary in the infinite population setting.

We finish this introduction with an outline of the present work.
In Section \ref{sec:definitions}, we introduce our main definitions, including the Moran process, the replicator equation, and the fitness potential, as well as their relation to game theory and their relevant properties. We also present the scaling assumptions, demonstrating the formal compatibility of these two models, see also Appendix~\ref{ap:formal}.
In Section \ref{sec:twocases}, we find asymptotic expressions for two fundamental examples of the Moran process: the constant fitness case, which corresponds to a one-player game, and the coordination case from two-player game theory.  
These results are foundational in Section \ref{sec:dplayer}, where we provide an asymptotic expression for the fixation probability in the limit of large effective population sizes of general potentials, which includes $d$-player games as particular cases. In particular, we show that the fixation probability approaches a convex combination of one-player dominance games and two-player coordination games, allowing a full characterization of Moran process and replicator system that are compatible. As a consequence, we will find a definition of finite population metastability that generalizes the well-known $\mathrm{ESS}_N$ condition.
All these concepts will be explored in the examples section, Section \ref{sec:examples}. In Section \ref{sec:inverse}, we will demonstrate how to construct a potential such that its associated Moran process is compatible with a given replicator equation. 
Our conclusions will be presented in Section \ref{sec:conclusions}.

\section{Processes in population genetics}
\label{sec:definitions}

The Moran process is a birth-and-death process, in which at each time step one individual is selected to die with equal probability, and one is selected to reproduce, possibly the same one, with probability proportional to a given function called \emph{fitness}~\cite{Moran,Nowak:06,Taylor_etal:BMB2004}. More precisely,
\begin{definition}
    Consider a population of size $N$, composed of individuals of two types, \A and \B. The state of the population is given by the presence of the focal type \A, $i\in\{0,\dots,N\}$. Let  $\Psi^{(\A)}_N,\Psi^{(\B)}_N:\{0,\dots,N\}\to\R_+$ be the fitness functions of type \A and \B, respectively. Whenever possible, we will omit the subindex $N$. The transition matrix from state $i$ to $j$ of the Moran process is given by
    \[
    M_{ij}=\left\{\begin{array}{ll}
    \frac{i(N-i)\Psi^{(\A)}(i)}{N(i\Psi^{(\A)}(i)+(N-i)\Psi^{(\B)}(i))} \ ,&\quad j=i+1\ ,\\
    \frac{i(N-i)\Psi^{(\B)}(i)}{N(i\Psi^{(\A)}(i)+(N-i)\Psi^{(\B)}(i))} \ ,&\quad j=i-1\ ,\\
    1-M_{i,i+1}-M_{i,i-1}\ ,&\quad j=i\ ,\\
    0\ ,&\quad\text{otherwise}\ .
    \end{array}\right.
    \]
    From now on, we refer to $i$ as the \emph{state} of the population.
\end{definition}

In the long run, the population will be homogeneous, and one of the central questions in evolutionary dynamics is to obtain the probability $F_i$ that the population consists eventually only of elements of the focal type, given that its initial state is $i$, the so called \emph{fixation probability}. For the Moran process, there is a well-known closed-form expression~\cite{AntalScheuring_06,Nowak:06}:
\begin{equation}\label{eq:Fk}
F_0=0\ ,\qquad F_i=\frac{\sum_{j=1}^i\prod_{k=0}^{j-1}\rho_k^{-1}}{\sum_{j=1}^N\prod_{k=0}^{j-1}\rho_k^{-1}}=
\frac{1+\sum_{j=1}^{i-1}\prod_{k=1}^j\rho_k^{-1}}{1+\sum_{j=1}^{N-1}\prod_{k=1}^i\rho_k^{-1}}\ ,\quad i=1,\dots,N\ ,
\end{equation}
where the \emph{relative fitness} is given by
\[
\rho_{i}\bydef\frac{\Psi^{(\A)}(i)}{\Psi^{(\B)}{(i)}}\ .
\]
We use the convention that $\sum_{i=1}^0=0$.

The fact that $F_0=0$ and $F_N=1$ highlights that the Moran process does not include mutations. 
The second expression for $F_i$ makes it explicit that the fixation probability does not depend on the values of $\rho_0$ and $\rho_N$, i.e., fitnesses are defined arbitrarily in homogeneous populations.

One important assumption in evolutionary dynamics is the so-called \emph{weak-selection principle}, which states that the difference in fitness between the two types only becomes apparent in evolution if the observation time is sufficiently long. The precise connections between the process time step, the population size, and how the fitness function approaches one as the population increases constitute the \emph{scaling relations} and are critical to mathematically defining the large population limit. 

\begin{definition}\label{def:wsp}
We say that the family of function $\Psi^{(\A)}_N$ and $\Psi^{(\B)}_N$ satisfy the \emph{weak selection principle} if there is a continuous functions $\omega:\R\to\R$, such that $\lim_{x\to\infty}\omega(x)=0$, $N_0\in\N$, and there are bounded continuous functions $\psi^{(\A)},\psi^{(\B)}:[0,1]\to\R$ such that
\begin{equation}\label{eq:wsp}
\Psi^{(\X)}_N(i)=1+\omega\left(N\right)\psi^{(\X)}\left(\frac{i}{N}\right),\quad i=1,\dots,N-1\ ,\quad N\ge N_0\ ,\quad \X=\A,\B\ .
\end{equation}

We also define the \emph{intensity of selection}, or the \emph{inverse of the effective population size}, for the finite population Moran process $\kappa_N$ as

\begin{equation}\label{eq:kappaN}
\kappa_N\bydef\frac{2\|\psi^{(\A)}-\psi^{(\B)}\|_\infty}{N\displaystyle\max_{i\in\{0,1,\dots,N\}}|\log\rho_i|}\ , 
\end{equation}
where $\|f\|_\infty=\sup_{x\in(0,1)}|f(x)|$ is the $L^\infty$ norm of $f$ in $(0,1)$.
\end{definition}

Eq.~\eqref{eq:kappaN} follows from~\cite{ChalubSouza:JMB2016}, using a representation more akin to the use of the intensity of selection in the biological literature, and more convenient to the purposes of the present work. See~\cite{ChalubSouza:JMB2016} for further details.

For large $N$, the intensity of selection can be approximated by
\[
\kappa_N\approx\frac{2\|\psi^{(\A)}-\psi^{(\B)}\|_\infty}{N\max_i[\omega(N)(\psi^{(\A)}(i/N)-\psi^{(\B)}(i/N))]}=\frac{2}{N\omega(N)}\ .
\]

\begin{remark}\label{rmk:scalings}
    In~\cite{ChalubSouza:JTB2018}, parameters $\mu$ and $\nu$ were defined to consider the large population limit, one for the intensity of the selective term, $\omega(N)=\left(\Delta t\right)^{\nu}$, and a second one for the scaling relation $N^{-1}=\varepsilon\left(\Delta t\right)^\mu/2$, where $\Delta t$ is the typical renewal time for one individual and $\varepsilon>0$ is a immaterial constant. The replicator equation is the large population limit for $\nu=1-\mu$, $\mu\in(1/2,1)$. Therefore,   $\omega(N)=\left(\frac{2}{\varepsilon}\right)^{\nu/\mu}N^{-\nu/\mu}\propto N^{-(\mu^{-1}-1)}$, where the symbol ``$\propto$'' indicates ``proportional to''. See Appendix~\ref{ap:formal} for further details.
\end{remark}

In the sequel, we will assume that
\begin{equation}\label{eq:alpha}
\omega(N)=N^{-\alpha}\ ,\quad \alpha\in(0,1)\ ,
\end{equation}
what was termed in~\cite{ChalubSouza:JMB2016} a \emph{selection driven} regime. In this case, $\kappa_N\approx 2 N^{\alpha-1}\to0$ in the limit $N\to\infty$. Cases that will not be considered here include $\alpha>1$, the \emph{quasi-neutral} regime, and the boundary case $\alpha=1$ in which both the natural selection and the genetic drift are present in the \emph{balanced} regime. This last is equivalent to the use of the Kimura equation as the large population, cf.~\cite{ChalubSouza:TPB2009} for the Moran process  and~\cite{Chalubsouza:JMB2014} for the Wright-Fisher process, with different scaling relations. 

A convenient tool used in the foregoing analysis is the \emph{fitness potential} $V:[0,1]\to\R$, introduced in~\cite{ChalubSouza:JTB2018}, as a useful quantity to describe the short and long-term behavior of the system, in a similar way to the use of potentials in classical mechanics. Namely
\begin{equation}\label{eq:def_V}
V(x)\bydef-\int_0^x\left(\psi^{(\A)}(s)-\psi^{(\B)}(s)\right)\rd s\ .
\end{equation}
By definition, $V(0)=0$, but nothing will change if we add an arbitrary constant to $V$.

An important result is a straightforward adaptation of \cite[Theorem 1]{ChalubSouza:JMB2016}:
\begin{theorem}\label{thm:continuousfixation}
 Let    
\begin{equation}\label{eq:varphikappa}
\varphi_\kappa[V](x)\bydef\frac{\int_0^x \exp\left(\frac{2}{\kappa}V(s)\right)\,\rd s}{\int_0^1 \exp\left(\frac{2}{\kappa}V(s)\right)\,\rd s}\ ,
\end{equation}
where $V\in C^3([0,1])$  is a given fitness potential.
Assume that a finite population of size $N$ evolves according to the Moran process with fitness $\Psi^{(\A)}$ and $\Psi^{(\B)}$ that satisfies the weak selection assumption. If $1\gg\kappa_N\gg N^{-1/2}$ with $N$ such that $N\omega(N)\approx \frac{2}{\kappa_N}$, which implies $N\omega(N)^2\approx\frac{4}{N\kappa_N^2}\ll1$, then the fixation probability vector can be well approximated by $\varphi_{\kappa}[V]$ at the uniform discretization points. More precisely,
\[
    F_i = F_{i,N}=\varphi_{\kappa_N}[V](i/N) + \order{N^{-1}\kappa_N^{-2}}\ ,\quad i=0,1,\dots, N\ .
\]
\end{theorem}

If, instead of assuming weak selection in the form of Eq.~\eqref{eq:wsp}, we impose $\Psi^{(\X)}=\exp\left(\omega(N)\psi^{(\X)}\right)$, $\X=\A,\B$, a perturbation $\order{\omega(N)^2}$ of the previous expression, then the error term improves considerably and becomes $\order{N^{-2}\kappa_N^{-1}}$, if $V$ is an interior potential (i.e., if its maximum is attained only at the interior) and $\order{N^{-2}\kappa_N^{-2}}$ otherwise.

One of the questions that motivated the present work is to find precise conditions such that finite population and infinite population descriptions of the evolutionary process are consistent. In an infinite population setting, the most celebrated model is the replicator equation 
\begin{equation}\label{eq:replicator}
\dot X=-X(1-X)V'(X)\ ,
\end{equation}
where $X\in[0,1]$ is the fraction of individuals of focal type. See~\cite{HofbauerSigmund} for further details on the history and motivations of this particular model.

As a consequence of the absence of mutation, homogeneous populations, identified by $X=0$ and $X=1$ are always stationary solutions of the replicator equation. However, for different fitness potentials $V$, mixed stable equilibria appear at $X_*$ whenever $V'(X_*)=0$, $V''(X_*)>0$ for $X_*\in(0,1)$. 

Assuming that the replicator equation models the first time scale evolution of the Moran process, then we expect that for initial conditions in the basin of attraction of $X_*$, the Moran process will initially evolve towards $X_*$ and then will be absorbed by one of the boundaries. Therefore, we expect that the fixation probability of the Moran process $\varphi_{\kappa_N}(x)$ will weakly depend on $x$, as long as it is in the basin of attraction of the same point.

Following these ideas, we define

\begin{definition}
    \label{def:compatibility}
Let $\Psi^{(\A)}$, $\Psi^{(\B)}$ be fitness functions and consider the associated replicator equation and Moran process, in the weak selection approximation, with given scaling parameter $\alpha$, cf. Eq.~\eqref{eq:alpha}. Let $X_{X_0}(t)$ be the solution of the Replicator Equation~\eqref{eq:replicator} with initial condition $X(0)=X_0$. 
We say that both models are compatible if there is a non-decreasing function $f:[0,1]\to[0,1]$ such that $f(0)=0$, $f(1)=1$, $f(x)\in(0,1)$ for all $x\in(0,1)$ and for all $X_0\in[0,1]$ it is true that
\begin{equation}\label{eq:dominance_repphi}
 X_{X_0}(t)\xrightarrow{t\to\infty}X_\infty\Longleftrightarrow \varphi_{\kappa_N}(X_0)\xrightarrow{\kappa_N\to0} f(X_\infty)\ .
\end{equation}
\end{definition}

Finally, we introduce a class of fitness potentials of interest and two measures of $V$ that will be useful in the sequel.
\begin{definition}\label{def:Vgeneric}
    Let $V:[0,1]\to\R$ be a fitness potential that is at least twice differentiable. We say that $V$ is \emph{generic} if
    \begin{enumerate}
    \item The fitness difference is non-degenerated on the boundaries, i.e., $V'(0), V'(1)\ne0$.
    \item All zeros of $V'$ are non-degenerated, i.e, for all $x\in(0,1)$ such that $V'(x)=0$, then $V''(x)\ne0$.
    \end{enumerate}
    Let $\mathcal{M}\bydef\{x_0,x_1\dots,x_M\}\subset[0,1]$ be the set of $M+1$ strict local maxima of $V$ such that $i<j$ implies $x_i< x_j$. 
It is clear that $\mathcal{M}\ne\emptyset$. We define two different measures on $V$: if $\mathcal{M}$ is a singleton, then we declare  $\Vgap{V}\bydef+\infty$ and $\Msep{V}\bydef1$; otherwise:
\begin{enumerate}
\item 
    The gap between the two largest local maxima of $V$, or the $V$-gap:
    \[
    \Vgap{V}\bydef \min\left\{V(x_i)-V(x_j)\bigl| x_i\in\mathop{\mathrm{arg\,max}}V(x), x_j\in\mathcal{M}\backslash\{x_i\}\right\}\ge 0\ .
    \]
\item The separation between the local maxima of $V$, or the $M$-sep as
\[
\Msep{V}\bydef \min_{i\ne j}|x_i-x_j|\in(0,1]\ .
\]
\end{enumerate}
\end{definition}

\section{Fixation probability formulas for two fundamental cases}
\label{sec:twocases}

Non degenerated two-player games fall in one of the following three categories: either i)  the fitness difference does not change sign in the interval $[0,1]$, with positive (negative) difference corresponding to dominance by \A (\B, respectively); ii) fitness difference is positive at $x=0$ and negative at $x=1$, indicating that both types are able to invade an homogeneous population of the other strategist type; iii) fitness difference is negative at $x=0$ and positive at $x=1$, indicating that both \A and \B are unable to invade an homogenous population of the other strategist. The first case is known as \emph{dominance}, the second one as \emph{coexistence}, while the third one is the \emph{coordination} case~\cite{HofbauerSigmund}. As an abuse of language, we can define a one-player game, in which fitnesses are independent of the composition of the populations  --- frequency-independent fitness ---, i.e., they are constants, making this example a particular case of a dominance game.

In this section, we will study in detail the case of frequency-independent fitness and the coordination games. The results presented are not only interesting \textit{per se}, but they turn out to play an important role in the asymptotic analysis of the fixation probabilities of $d$-player games that we will be carried out in Sec.~\ref{sec:dplayer}.

We begin with the constant fitness case. Without loss of generality, we take $\Psi^{(\B)}_N\equiv1$ and hence, $\psi^{(\B)}\equiv0$; let $\psi^{(\A)}(x)=r$ and then we obtain that $V(x)=-rx$. On the other hand, let $\rho\bydef\frac{\Psi^{(\A)}(i)}{\Psi^{(\B)}(i)}=\Psi^{(\A)}(i)>0$.

From Eq.~\eqref{eq:Fk} and Def.~\ref{def:wsp}, we conclude that $\rho=1+\omega(N)r$ and, 
\[
\kappa_N=\frac{2r}{N\log\rho}=\frac{2r}{N\log(1+\omega(N)r)}\quad\Longleftrightarrow\quad
\omega(N)=\frac{1}{r}\left(\re^{\frac{2r}{\kappa_NN}}-1\right)\ .
\]
Therefore,
\[
F_{xN}=\frac{1-\rho^{-xN}}{1-\rho^{-N}}=\frac{1-(1+\omega(N)r)^{-xN}}{1-(1+\omega(N)r)^{-N}} {=}\frac{1-\re^{-2xr/\kappa_N}}{1-\re^{-2r/\kappa_N}}=\varphi_{\kappa_N}[V](x)\ .
\]

For $V(x)=-rx$, we introduce, for the sake of convenience,  the following notation:
\begin{align}
\label{eq:varphiA}
\varphi^{(\A)}_{\kappa}[r](x)&\bydef \varphi_{\kappa}[V](x)=\frac{1-\re^{-2xr/\kappa}}{1-\re^{-2r/\kappa}}\ ,\quad\text{if}\ r>0\ ,\\
\label{eq:varphiB}
\varphi_{\kappa}^{(\B)}[r](x)&\bydef 1-\varphi^{(\A)}_\kappa[-r](1-x)=\frac{\re^{-2|r|(1-x)/\kappa}-\re^{-2|r|/\kappa}}{1-\re^{-2|r|/\kappa}}\ ,\quad\text{if}\ r<0\ .
\end{align}
We shall see in Sec.~\ref{sec:dplayer} that these formulae continue to be valid for dominance games with any finite number of players. Eq.~\eqref{eq:varphiA} is a shorthand notation for the fixation probability of type \A, when \A is dominant (i.e., when $r>0$), while Eq.~\eqref{eq:varphiB} is a shorthand notation for the fixation probability of type \A when type \B is dominant, i.e., when $r<0$. The rationale behind Eq.~\eqref{eq:varphiB} is the following: the fixation probability of \A is equal to its extinction probability $1-\varphi_\kappa$ when we swap types \A and \B, and, therefore, $\Psi^{(\A)}$ and $\Psi^{(\B)}$, i.e, when we change $r$ to $-r$, and $x$ to $1-x$.

Before proceeding to the coordination case, we digress on the compatibility of one-player dominance games in the Moran and the replicator dynamics. On the one hand, the flow of the Replicator Equation~\eqref{eq:replicator} is such that for any initial condition $X_0\in(0,1)$, $X(t)\to 1$ or $X(t)\to0$ for dominance by \A or by \B, respectively. On the other hand, fixation of the dominant type is almost sure --- except when its initial presence is very small. Thus, we have that, for one-player dominance games, the replicator and the Moran process are compatible in the sense of Definition~\ref{def:compatibility}, with any continuous function $f$ in $[0,1]$ such that $f(0)=0$, $f(1)=1$. See also the left panel in Fig.~\ref{fig:basic_cases}.

For coordination two-player games, we recall the standard representation of a two-player game payoff matrix:  $\left(\begin{smallmatrix} a&b\\ c&d\end{smallmatrix}\right)$, with $a>c$ and $b<d$. We define $\gamma\bydef-\frac{a-c+d-b}{2}<0$ and $\xs\bydef\frac{d-b}{a-c+d-b}\in(0,1)$. The fitness potential is given by
\begin{equation}\label{eq:Vtwoplayer}
V(x)=\gamma(x-\xs)^2-\gamma(\xs)^2\ ,
\end{equation}
corresponding to the fitness difference of
\[
\Psi^{(\A)}(x)-\Psi^{(\B)}(x)=-V'(x)%=(a-c+d-b)\left(x-\frac{d-b}{a-c+d-b}\right)
%=(a-c+d-b)x-(d-b)
=(ax+b(1-x))-(cx+d(1-x))\ .
\]
We note that $V'(0)=d-b>0$,  $V'(1)=c-a<0$, and $V''(\xs)=2\gamma<0$.

The asymptotic expression of the fixation probability is given by $\varphi_\kappa(x)=\varphi_\kappa^{(\mathrm{C})}(x)+\order{\kappa_N^{2}}$, where
\begin{equation}
    \label{eq:varphiC}
    \varphi_\kappa^{(\mathrm{C})}[\xs,V''(\xs)](x)\bydef\frac{\Ncal\left(\sigma_\kappa^{-1}(x-\xs)\right)-\Ncal\left(-\sigma_\kappa^{-1}\xs\right)}{\Ncal\left(\sigma_\kappa^{-1}(1-\xs\right))-\Ncal\left(-\sigma_\kappa^{-1}\xs\right)},
\end{equation}
$\sigma_\kappa^{-1}\bydef\sqrt{\frac{2|V''(\xs)|}{\kappa}}$, and $\Ncal(x)\bydef\frac{1}{\sqrt{2\pi}}\int_{-\infty}^x\re^{-y^2}\,\rd y$ is the cumulative normal distribution.

This formula follows from a slightly modified form of Laplace's method for a function with unique interior maxima; cf. Appendix~\ref{ap:varphiC}; see also~\cite{ChalubSouza:JMB2016} where a similar formula, with exponential small errors at the boundaries, appears for the first time. For other cases in two-player games, we ask the reader to wait until the next section.

Regarding the compatibility in this case, notice that the Replicator flow will converge to $0$, if $X_0<\xs$, and to $1$, if $X_0>\xs$. For the Moran process, type \A will be extinct or fixated almost surely if its initial presence is less than or larger than $\xs$, respectively. See also the right panel on Figure~\ref{fig:basic_cases}.

We point out that  the cases studied so far have $\Vgap{V}=+\infty$ and $\Msep{V}=1$. Furthermore, the function $f$ in Def.~\ref{def:compatibility} is such that $f(\xs)=\frac{1}{2}$.

\begin{figure}
\includegraphics[width=0.49\textwidth]{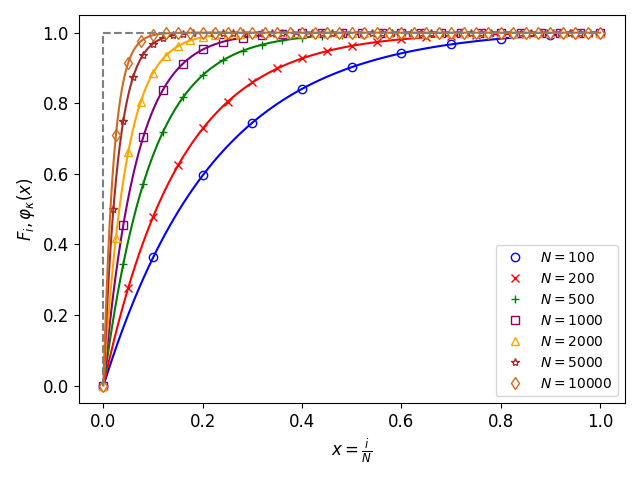}\hfill
\includegraphics[width=0.49\textwidth]{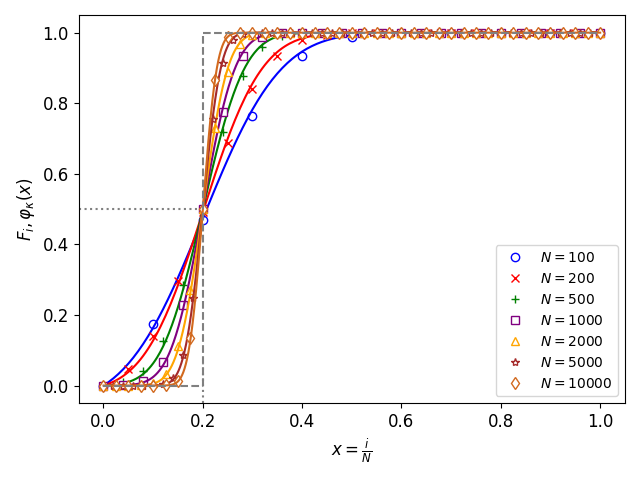}\\
\caption{Fixation probability for the Moran process (markers) and asymptotic approximation of $\varphi_{\kappa_N}$ for one-player game dominated by \A, with $\mathbf{a}^{(1)}=(1.5)$, $\mathbf{b}^{(1)}=(1)$ and $\alpha=1/2$ (left) and for two-player coordination game with $\alpha=1/3$ and payoffs given by $\left(\begin{smallmatrix}5&4\\ 1&5 \end{smallmatrix}\right)$, which corresponds, in the notation of Sec.~\ref{sec:dplayer}, to $\mathbf{a}^{(2)}=(4,5)$, $\mathbf{b}^{(2)}=(5,1)$, i.e.,
$\gamma=-5$ and $\xs=1/5$ (right). Note that at the unstable equilibrium of the replicator equation in the right picture, $\xs=1/5$, the fixation probability is $1/2$ (dotted line), indicating equal probability of moving to the left or to the right from a state initially at $\xs$. The color line indicates $\varphi_{\kappa_N}$ given by Eq.~\eqref{eq:varphiA} (left) and Eq.~\eqref{eq:varphiC} (right), for the same value of $N$ as the markers. For one-player and two-player games, the value of $\kappa_N$ given by Eq.~\eqref{eq:kappaN} simplifies to $\kappa_N=\frac{2|\rho-1|}{N\left\|\log\frac{1+\omega(N)\rho}{1+\omega(N)}\right\|}$ and 
$\kappa_N=\frac{2\max\{|a-c|,|b-d|\}}{N\max\left\{\left|\log\frac{1+\omega(N)a}{1+\omega(N)c}\right|,\left|\log\frac{1+\omega(N)b}{1+\omega(N)d}\right|\right\}}$, respectively. The dashed black line indicates the $t\to\infty$ limit of the Replicator Equation~\eqref{eq:replicator} as a function of the initial condition $x$, and also corresponds to $\lim_{\kappa\to0}\varphi_\kappa$, showing that the Moran process and the replicator equation are compatible in both cases with the function $f$ in Def.~\ref{def:compatibility} such that $f(0)=0$, $f(1)=1$ (both) and $f\left(\frac{1}{5}\right)=\frac{1}{2}$ (right). For clarity, not all values of $F_i$ are plotted.}
 \label{fig:basic_cases}   
\end{figure}

\section{Asymptotic expression for general $d$-player games}
\label{sec:dplayer}

We are ready to prove our main theorem, namely that the asymptotic expression of the fixation probability of the Moran process for a generic fitness potential $V$ is a convex combination of the fixation probability of one-player dominance games and two-player coordination games. The latter will be associated with local maxima of the fitness potential on the boundaries, while the former with interior local maxima of $V$.
\begin{theorem}\label{thm:asymptotic}
    Let $V\in C^3([0,1])$ be according to Def.~\ref{def:Vgeneric}, and let $\mathcal{M}=\{x_0,\dots,x_M\}$ be the set of local maxima of $V$. Let $\kappa=\kappa_N$, given by Eq.~\eqref{eq:kappaN}, and for $x_i\in(0,1)$ define $\sigma_{\kappa,i}^{-1}\bydef\sqrt{2|V''(x_i)|/\kappa}$, $i=0,\dots,M$. Assume further that $\Msep{V}\gg\kappa$ and that if $x_0\not=0$ then $x_0\gg\kappa$ and if $x_M\not=1$ then $1-x_M\gg\kappa$, which is always the case if $\kappa$ is sufficiently small. Then, there are parameters $c_k$, $k=0,\dots,M$ such that 
\begin{displaymath}
  \phi_\kappa(x)=\sum_{i=0}^Mc_i\varphi_{i,\kappa}(x) + \order{\kappa}\ ,\quad 0<\kappa\ll 1\ ,  
\end{displaymath}
where
\begin{align*}
x_0=0&\Rightarrow\varphi_{0,\kappa}(x)=\varphi_\kappa^{(\A)}[V'(0)](x)\ ,\qquad \text{cf. Eq.~\eqref{eq:varphiA}}\ ,\\
x_M=1&\Rightarrow\varphi_{M,\kappa}(x)=\varphi_{\kappa}^{(\B)}[V'(1)](x)\ ,\qquad \text{cf. Eq.~\eqref{eq:varphiB}}\ ,\\
x_i\in(0,1)&\Rightarrow\varphi_{i,\kappa}(x)=\varphi^{(C)}_\kappa[\xs,V''(\xs)](x)\ ,\quad i=0,\dots,M\ ,\qquad \text{cf. Eq.~\eqref{eq:varphiC}}\ .
\end{align*}
Furthermore
\[
c_i\bydef\frac{d_i}{\sum_{i=0}^Md_i}\ ,
\]
where
\begin{align*}
    x_i=0,1&\Rightarrow d_i\bydef\frac{\kappa\exp\left(2V(x_i)/\kappa\right)}{2|V'(x_i)|}\left(1-\exp\left(-2|V'(x_i)|/\kappa\right)\right)%\frac{\re^{2V(1)/\kappa}}{V'(1)}\left(1-\re^{-2V'(1)/\kappa}\right)
    , \\
    x_i\in(0,1)&\Rightarrow    d_i\bydef{\sigma_{\kappa,i}}\sqrt{2\pi}\exp\left(2V(x_i)/\kappa\right)\left(\Ncal(\sigma_{\kappa,i}^{-1}(1-x_i))-\Ncal(-\sigma_{\kappa,i}^{-1}x_i)\right).
\end{align*}    
\end{theorem}

\begin{proof}
    See Appendix~\ref{ap:ThmAsymptotic}.
\end{proof}

With these results at hand, and using $\approx$ to indicate leading order terms, we finish the formulas for two-player games. The case in which \A dominates \B, using two-player game theory parlance, is associated with a monotonically decreasing $V$, which has the only local maximum at $x=0$. Therefore, $\varphi_{\kappa}(x)\approx\varphi_{\kappa}^{(\A)}[V'(0)](x)$, $\kappa=\kappa_N$. If \B dominates \A, remembering that $\mathcal{M}=\{1\}$, we conclude that $\varphi_{\kappa}(x)\approx\varphi_{\kappa}^{(\B)}[V'(1)](x)$. 

Let us study coexistence games in more detail. In this case, the best response for each player is to adopt a strategy that is different from the one adopted by the adversary. The payoff matrix is such that $\left(\begin{smallmatrix}a&b\\ c&d\end{smallmatrix}\right)$, with $a<c$ and $b>d$. The fitness potential is given by Eq.~\eqref{eq:Vtwoplayer}, $\gamma>0$ and $x_*\in(0,1)$. Furthermore, $\mathcal{M}=\{0,1\}$. The $V$-gap is given by $\Vgap{V}=\gamma|1-2\xs|$, and the asymptotic expression of the fixation probability is a convex combination of $\varphi_{\kappa_N}^{(\A)}[V'(0)]$ and $\varphi_{\kappa_N}^{(\B)}[V'(1)]$.

If $\Vgap{V}=0$, which is true if and only if $\xs=1/2$, then $V(0)=V(1)$ and 
\[
\varphi_{\kappa_N}(x)=\frac{1}{2}\left(\varphi_{\kappa_N}^{(\A)}[V'(0)](x)+\varphi_{\kappa_N}^{(\B)}[V'(1)](x)\right)+\order{\kappa_N}\ . 
\]
For $\Vgap{V}>0$, $\varphi_{\kappa_N}$ approaches asymptotically one of the dominance expressions as $\kappa_N\to 0$, namely $\varphi_{\kappa}\approx\varphi_{\kappa}^{(\A)}$ if $\xs>1/2\Leftrightarrow V(1)<V(0)$ and $\varphi_{\kappa}\approx\varphi_{\kappa}^{(\B)}$ if $\xs<1/2\Leftrightarrow V(1)>V(0)$. However, in any case, the replicator equation is such that an interior stable equilibrium exists. The compatibility between large, but finite and infinite population models for coexistence games is limited to the case in which $\xs=1/2$. In all the other cases, non-trivial fixation plateaux (i.e., a set of initial conditions such that the fixation probability is essentially constant and away from 0 and 1) are transient and disappear as $N\to\infty$ or, equivalently, $\kappa_N\to0$. 

We conclude that non-trivial fixation plateaux may exist in the limit $N\to\infty$ only if $\Vgap{V}=0$, otherwise, they are transient, as stated in the next result.
\begin{corollary}\label{cor:asymptotic}
If $\Vgap{V}>0$, then the fixation pattern in the limit $\kappa_N\to0$ is given either by the fixation pattern of a one-player game or by the fixation pattern of a two-player coordination game.    
\end{corollary}

\begin{proof}
If $\Vgap{V}>0$, then $V$ has a unique global maximum. In this case, when $\kappa_N\to0$, all but one of the $c_i$ coefficients will vanish in the limit, cf. Thm.~\ref{thm:asymptotic}.
\end{proof}

For any initial condition $x\in[0,1]$ in one-player and $x\in[0,1]\backslash\{\xs\}$ in two-player coordination games, the fixation probability is such $\lim_{\kappa\to0}\varphi_{\kappa}(x)\in\{0,1\}$; in the latter case, we also have that $\lim_{\kappa\to0}\varphi_\kappa(\xs)=\frac{1}{2}$. In these two cases, these are the only additional constraints to the function $f$ in Def.~\ref{def:compatibility}. In general, the choice of the function $f$ is constrained only by its values on the unstable points of the replicator equation and on the stable points that still exist after taking the large population limit. From the assumptions, these points are finite, and therefore, there is a large degree of freedom in defining $f$. On one hand, a value of 0 (1) corresponds to extinction (fixation, respect.) in the Moran process and to attraction by 0 (1, respect.) in the replicator dynamics. On the other hand, a value between 0 and 1 indicates the existence of a region of constant fixation probability, or a \emph{fixation plateau}, in the large population limit of the Moran process, and to at least one attractor of the replicator dynamics in the region of the plateau. According to Thm.~\ref{thm:asymptotic} and Corollary~\ref{cor:asymptotic}, the former is sensitive to small perturbations of $V$, but the latter is not. A fixation plateau that exists even after taking the large population limit in the Moran process represents an interior attractor of the initial dynamics of the Moran process, before absorption, that we term an \emph{essentially metastable state}, or, in short, an EMS. In the sequel, we are able to fully characterize its existence according to the properties of the fitness potential $V$. However, before proceeding and following~\cite{ChalubSouza:JTB2018}, we introduce the $\frac{\kappa}{2}$-exponential average of $V$:
\[
\langle V\rangle_\kappa\bydef\frac{\kappa}{2}\log\int_0^1\re^{2V(s)/\kappa}\rd s\ .
\]

\begin{definition}\label{def:EMS}
Let $\xmin\in(0,1)$ be a local minimum of a generic fitness potential $V$. Define $I_\kappa$, $\kappa>0$, as the connected component of the set $\{x\in[0,1]| V(x)<\langle V\rangle_\kappa\}$ such that $\xmin\in I_\kappa$.
    Then, $\xmin$ is an EMS if there is $\kappa_0>0$ such that for any $0<\kappa<\kappa_0$ $I_\kappa\ne\emptyset$, and $0,1\not\in I_\kappa$.
\end{definition}

Other local minima of $V$ will be referred to as metastable states; they are not essential in the sense that they vanish when $N\to\infty$.

It turns out that we can completely characterize the generic potentials that admit EMSs. In order to do so, we first recall a slightly modified form of a definition made in \cite{ChalubSouza:JMB2016}:
\begin{definition}
    We shall say that a fitness potential $V$ is an \textbf{interior potential} is global maxima is only attained interior. Otherwise it will be termed a \textbf{boundary potential}. We shall also say that $V$ is a \textbf{pure boundary potential}, if its global maxima is only attained at the boundaries.  
\end{definition}

\begin{theorem}\label{thm:plateau}
    Let $V\in C^3([0,1])$ be a generic fitness potential. Then, an EMS exists if and only if $\Vgap{V}=0$ and $V$ is either an interior potential, or a pure boundary potential, or a boundary potential with at least two interior global maxima. Furthermore, if $\xmin$ is an EMS, then the limit $\kappa\to0$ of the fixation probability has a plateau over a set that contains the basin of attraction of $\xmin$ in the replicator dynamics, given by Eq.~\eqref{eq:replicator}. More generally, let $x_-<x_+$ denote two consecutive global maxima. If $\xmin^1,\ldots,\xmin^k \in (x_-,x_+)$, then the fixation probability will have a plateau over the closure of the union of their basin of attraction (the unstable points in between will also be included).
\end{theorem}

\begin{proof}
    If $\Vgap{V}>0$, then $V$ has a unique global maximum on the unit interval. Therefore, for any local minimum, the corresponding $I_\kappa$ will contain at least one of the boundary points, provided we take $\kappa$ sufficiently small. This shows the necessity of the gap condition. For sufficiency, let $\Vgap{V}=0$ and assume for the time being that $V$ is not a pure boundary potential; in this case, note that, since $V$ is generic, there exists at least one and at most finitely many local minima in $(x_-,x_+)$.  Since $\lim_{\kappa\to0}\langle V\rangle_\kappa=\sup_{x\in[0,1]}V(x)$,  there exists a  sufficiently small $\kappa$ and an open interval $I'_\kappa\subset(x_-,x_+)$  such that  for any local minimum  $\xmin\in (x_-,x_+)$, $I'_\kappa = I_\kappa$.  Thus, any $\xmin \in (x_-,x_+)$ is an EMS. The second claim follows from Theorem~\ref{thm:asymptotic} and on verifying that $\varphi_{\kappa_N}(x)-\varphi_{\kappa_N}\left(x_-+\order{\kappa_N^{1/2}}\right)\leq \order{\kappa_N^{1/2}}$, for $x<x_+-\order{\kappa_N^{1/2}}$. Hence the fixation probability has a plateau on $(x_-,x_+)$. If $V$ is a pure boundary potential then $x_-=0$ and $x_+=1$, and the proof is similar except that the error estimates are now of $\order{\kappa_N}$. Finally, if $\Vgap{V}=0$ and $V$ is a boundary potential without two interior global maxima, then term in the convex combination of $\phi_\kappa$ corresponding to the boundary term maximum will be $\order{\kappa_N^{1/2}}$ whereas the one corresponding to interior maximum will be $\order{1}$ and thus the former will vanish in the limit $\kappa_N\to0$.
\end{proof}

An immediate consequence of the previous theorem is that in two-player game theory (quadratic fitness potential $V$ given by Eq.~\eqref{eq:Vtwoplayer}, with $\gamma\ne0$ and arbitrary $\xs$), EMSs exist if and only if $\gamma>0$ and $\xs=1/2$ --- cf. the $\frac{1}{2}$-law in \cite{ChalubSouza:JMB2016}. In particular, the only finite but large population dynamics that are modelled after two-player games that are compatible with the replicator equation are games of dominance, coordination and the special case of coexistence in which the evolutionarily stable strategy, or ESS, is exactly at $x=1/2$ (see, e.g.,~\cite{HofbauerSigmund} for the definition of ESS).

We conclude this section by showing that Def.~\ref{def:EMS} is a natural extension of the classical $\mathrm{ESS}_N$ for $N$ large enough, as introduced in~\cite{Nowak2004}. We say that \B is an $\mathrm{ESS}_N$ if i) selection opposes \A invading \B, i.e., $\Psi^{(\A)}(1)<\Psi^{(\B)}(1)$, and ii) selection opposes \A replacing \B, i.e., $F_1<1/N$. Assume $\xmin=0$ is a local minimum of a generic $V$ such that $1\not\in I_\kappa\ne\emptyset$ for $\kappa$ sufficiently small. Initially, note that $V'(0)>0$ and, therefore, for sufficiently large $N$
\[
\Psi^{(\A)}(1)-\Psi^{(\B)}(1)=-\omega(N)V'(1/N)<0\ .
\]
On the other hand, $V(0)<\langle V\rangle_\kappa$ for $\kappa$ sufficiently small, and on differentiating Eq.~\eqref{eq:varphikappa}, we conclude that for $N$ sufficiently large
\[
\varphi'_{\kappa}(x)=\re^{\frac{2}{\kappa}\left(V(x)-\langle V\rangle_\kappa\right)}<1\ , \quad\forall x\in[0,1/N]\ .
\]
We conclude from Thm.~\ref{thm:asymptotic} that $F_1\approx\varphi_{\kappa}(1/N)<1/N$. A similar argument shows the relationship between dominance by \A and essential metastability of the point $x=1$.

\begin{remark}
In order to include the boundary values $x=0,1$ in Def.~\ref{def:EMS} we need to consider $\xmin\in[0,1]$ and the final condition should be rewritten as $\left\{\left\{0,1\right\}\backslash\{\xmin\}\right\}\cap I_\kappa=\emptyset$. For clarity, we opted to keep Def.~\ref{def:EMS} slightly more restrictive than we could and considered the boundary cases separately. Note that the existence of a plateau including $x=0$ ($x=1$) is associated with the extinction (fixation, respect.) of focal type \A. However, the use of the expression \emph{metastability} in these cases seems inappropriate, as these are stable points of the Replicator Equation~\eqref{eq:replicator} that correspond to absorbing states of the Moran process.
\end{remark}

\section{Examples}
\label{sec:examples}

The study of games with an arbitrary number of players, known as $d$-player games, dates back at least to the seminal work by Nash~\cite{Nash1950-lv}. Recent developments have considered the possibility that the number of players in a given game depends on the available strategies~\cite{HansenChalub:JTB2024}. Nevertheless, evolutionary game theory has mainly considered two-player games, either to highlight pairwise contests, which are important for population structured as graphs~\cite{Lieberman2005,Szabo2007-bx}, or to consider mean-field approximation~\cite{M-S1982, PriceM-S1973}, in which each individual plays against an average player --- cf. also \cite{Tarnita2009-gr}. However, in the study of social evolution the value of $d$ is an important modeling parameter~\cite{GokhaleTraulsen_DGA2014, GokhaleTraulsen:PNAS2010,KurokawaIhara2009,PSSS2009, SPS2009}. See~\cite{MurgelSouza2025} for a discussion on how certain central concepts of evolutionary game theory depend on the value of $d$.

In our setting, we assume that, for a large, but finite, population with a fraction $x\in[0,1]$ of individuals of type \A and $1-x$ of type \B, the average fitness of both types is given asymptotically in $N$ by~\cite{Lessard_2011,GokhaleTraulsen_DGA2014,GokhaleTraulsen:PNAS2010}
\[
\psi^{(\A)}(x)\bydef\sum_{i=0}^{d-1}\binom{d-1}{i}x^i(1-x)^{d-1-i}a_i^{(d)}\ ,\qquad
\psi^{(\B)}(x)\bydef\sum_{i=0}^{d-1}\binom{d-1}{i}x^i(1-x)^{d-1-i}b_i^{(d)}\ ,
\]
where $a^{(d)}_i$ and $b^{(d)}_i$ are the payoffs of type \A and type \B individuals, in a $d$-player game in which there are $i$ opponents of type \A. We also define $\mathbf{a}^{(d)}\bydef(a_i^{(d)})_{i=0,\dots,d-1}$, $\mathbf{b}^{(d)}\bydef(b_i^{(d)})_{i=0,\dots,d-1}$. 

The fitness potential is given as a function of the game payoff by
\[
V(x)=-\frac{1}{d}\sum_{i=0}^{d-1}(a^{(d)}_i-b^{(d)}_i)\mathcal{B}_{(i+1,d-i)}(x)\ ,
\]
where 
\[
\mathcal{B}_{(i,d)}(x)\bydef\frac{\Gamma(i+d)}{\Gamma(i)\Gamma(d)}\int_0^xt^{i-1}(1-t)^{d-1}\rd t
\]
is the \emph{incomplete beta function}~\cite{Gradshteyn}. Fitnesses are not uniquely determined by the potential, but if $V$ is a polynomial of degree $d$, $\mathbf{a}^{(d)}$ and $\mathbf{b}^{(d)}$ can be obtained as the result of an interaction between individuals modeled by $d$-player games. In the sequel, whenever we chose $\mathbf{a}^{(d)}$ and $\mathbf{b}^{(d)}$ from a given potential, we consider $b^{(d)}_i$ to be independent of $i$ and such that $a^{(d)}_i,b^{(d)}_i>0$. Adding a constant to the payoffs does not change the potential, but it does change the finite population fixation probability --- cf. \cite{ChalubSouza:JMB2016}. However, these changes turn out to be negligible in the large population regime, since the effect of adding such a constant vanishes as $N$ grows large and since our results are asymptotic as $N\to\infty$ or, equivalently, as $\kappa_N\to0$~\cite{FaureSchreiber}.

In Fig.~\ref{fig:dplayer}, the $V$-gap is strictly positive, and therefore there is no EMS. In the left column, we show that $\lim_{\kappa\to0}\varphi_\kappa$ is of dominance type, even though the associated replicator equation has two interior stable equilibria. The fact that $\arg\max_{x\in[0,1]}V(x)=\{1\}$ implies the dominance of \B-type individuals, i.e., $\lim_{\kappa\to0}\varphi_{\kappa}(x)=0$ for all $x\ne1$.  In the right column, $\arg\max_{x\in[0,1]}V(x)=\{\xs\}$, with $\xs\in(0,1)$, indicating a coordination type game. In the limit as $\kappa\to0$, there is a discontinuity in the limit of $\varphi_{\kappa}$ at the unstable equilibrium of the replicator equation. These are direct numerical verifications of Corollary~\ref{cor:asymptotic}. The limit value of $\varphi_{\kappa}(x)$ can be concluded from the following argument. Consider the two local minima of $V$, $x=0$, and $x=x_{\min}\in(\xs,1)$. For each one, build the corresponding interval $I^0_{\kappa}$ and $I^1_{\kappa}$, such that, for all values of $\kappa$ it is true that $0\in I^0_{\kappa}$ and $1\in I^1_{\kappa}$. Furthermore, as $x\ne \xs$, for $\kappa$ sufficiently small, $x\in I^0_{\kappa}$ or $x^1_{\kappa}$. In the first case, fixation of type \B, corresponding to $x=0$, is more likely, while fixation of type \A, corresponding to $x=1$, is more likely in the second case. Note that $x\in I^0_{\kappa}$ ($x\in I^1_{\kappa}$), for $\kappa$ sufficiently small, if and only if $x<\xs$ ($x>\xs$, respect.), and they show clearly, that absorption from a generic initial condition $x\in(0,1)\backslash\{\xs\}$ is more likely to occur at 0 or at 1 depending on $x\in (0,\xs)$ or $x\in (\xs,1)$, respectively. Equivalently, this is the boundary that requires minimum energy to be reached, cf.~\cite{ChalubSouza:JTB2018}. We finish this example noting that function $f$ in Def.~\ref{def:compatibility} for the left column is such that $f(0)=0$, $f(1)=1$ and otherwise arbitrary; for the right column, additionally $f(\xs)=\frac{1}{2}$.

\begin{figure}
    \centering
\includegraphics[width=0.48\textwidth]{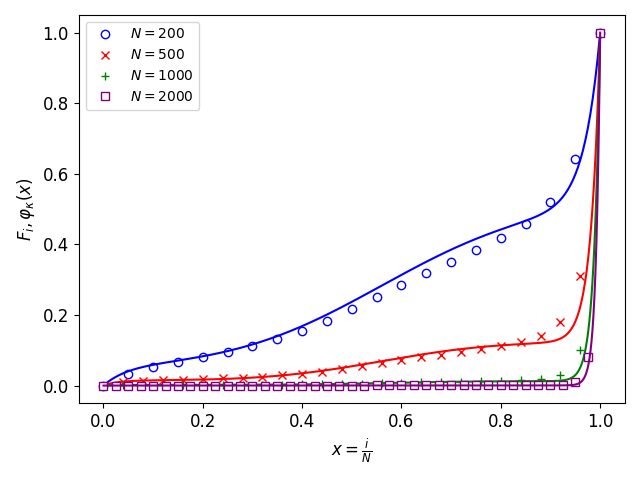}
\includegraphics[width=0.48\textwidth]{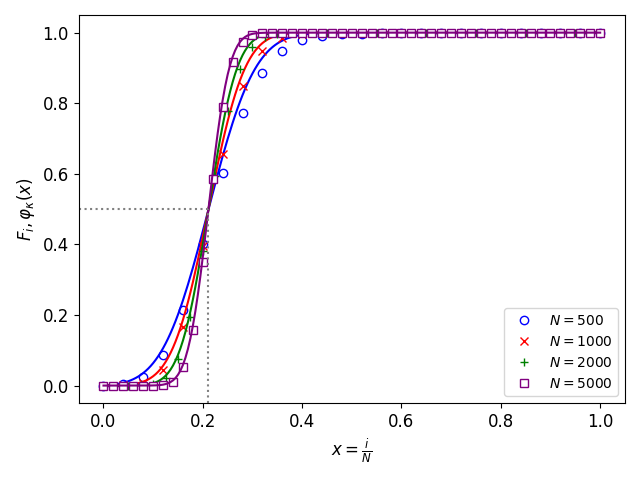}\\
\includegraphics[width=0.48\textwidth]{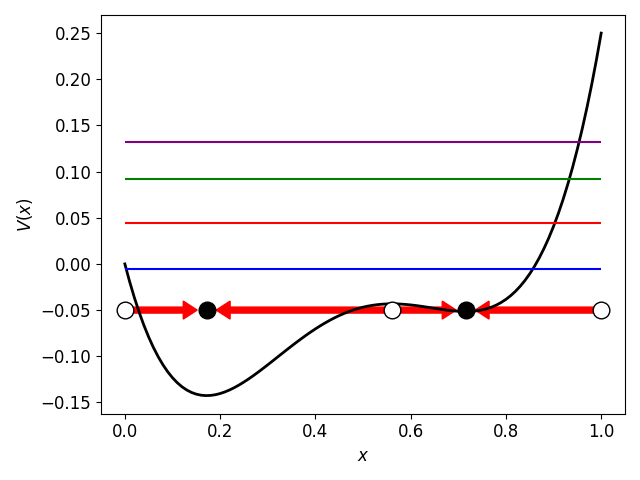}
\includegraphics[width=0.48\textwidth]{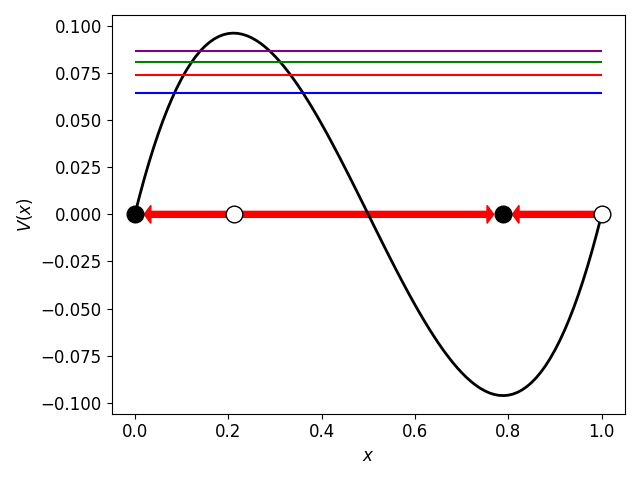}
\caption{Fixation probability for several values of $N$ (above), the associated fitness potential $V$, values of $\langle V\rangle_{\kappa_N}$, and the replicator dynamics flow, and where black and white circles mark stable and unstable points of the replicator dynamics, respectively (below). The use of colors is consistent among upper and lower figures. In the left column, payoffs are given by 
$\ba^{(4)} = (10, 3, 9, 4)$,
$\bb^{(4)} = (8, 7, 5, 7)$, $\alpha=1/2$
while in the right column, there are
$\ba^{(4)} = (1,3,3,1)$, and
$\bb^{(4)} = (2,2,2,2)$, with $\alpha=2/3$.
Fitness potentials $V$ are given by $x(x-1)(7x^2-7x+2)$ and by $x(2x-1)(x-1)$, respectively.
The left case approaches a dominance by \B, while the right case approaches a coordination case in which the discontinuity is associated with the only unstable equilibrium of the replicator equation in $(0,1)$. The fixation probability of the unstable equilibria of the replicator equation is $1/2$, see also Fig.~\ref{fig:basic_cases}. Stable equilibria are associated with trivial plateaux in both cases, indicating extinction and fixation, respectively. Consequently, there is no compatibility between the Moran process and the replicator equation in both cases.
In the left column, $f$ as in Def.~\ref{def:compatibility} is such that $f(0)=0$, $f(1)=1$; in the right column, additionally, $f(\xs)=\frac{1}{2}$.
}
\label{fig:dplayer}
\end{figure}

In Fig.~\ref{fig:plateau}, we show two cases in which the replicator dynamics and the Moran processes are compatible and there is at least one non-trivial fixation plateau. In the left figure, we consider $\mathbf{a}^{(3)}=(1,25,1)$ and $\mathbf{b}^{(3)}=(10,10,10)$, implying $V(x)=9x-24x^2+16x^3$, and therefore $\mathcal{M}=\{x_-,1\}$, with $x_-=\frac{1}{4}\in(0,1)$. The asymptotic behavior of the fixation probability is given $\varphi_{\kappa}(x)\to0$ for $x\in[0,x_-)$, and $\varphi_{\kappa}(x)\to \nu\approx 0.9314\in(0,1)$ for $x\in(x_-,1)$, indicating the existence of a non-trivial plateaux on the basin of attraction (in the Replicator Equation~\eqref{eq:replicator}) of $\xmin=\frac{3}{4}\in\left(x_-,1\right)$. The values of $\langle V\rangle_{\kappa_N}$ are indicated in the lower figure, showing that $\xmin$ is an EMS. Interestingly $\varphi_{\kappa}(x_-)=\nu/2$, indicating that at the unstable equilibrium $x_-$ the system moves with equal probability to the left and becomes extinct or to the right, approaching $\xmin$ and then fixating with probability $\nu$. In this case, $f(x_-)=\nu$ and $f(x_{\min})=1-\nu$, besides its boundary values. We considered $\alpha=4/5$; when $\alpha$ is closer to one, the convergence of $\varphi_\kappa$ to its limit is faster, and in the case in which there are plateaux, this is important as the value of $N$ required to observe the asymptotic behavior of $\varphi_{\kappa_N}$ is higher.

In the right column, we consider the case $V(x)= -2x+10x^2-16x^3+8x^4$, associated to 
$\mathrm{a}^{(4)}=(3, -11/3, 14/3, 1)$, and $\mathrm{b}^ {(4)}=(1,1,1,1)$, with $\alpha=2/3$. There are two stable attractors of the replicator equation at $x_\pm-=\frac{2\pm\sqrt{2}}{4}$ and unstable equilibria at $0$, $1/2$, $1$. There are two plateaux, one at level $\mu\approx 0.03695$ and the other one at $1-\mu$, which can be proved from the symmetry of the problem. In the same sense, $x_+=1-x_-\approx 0.146$. At the interior unstable equilibria, the fixation probability is $1/2$. Finally, $f(x_-)=\mu$, $f(x_+)=1-\mu$.

\begin{figure}
    \includegraphics[width=0.48\textwidth]{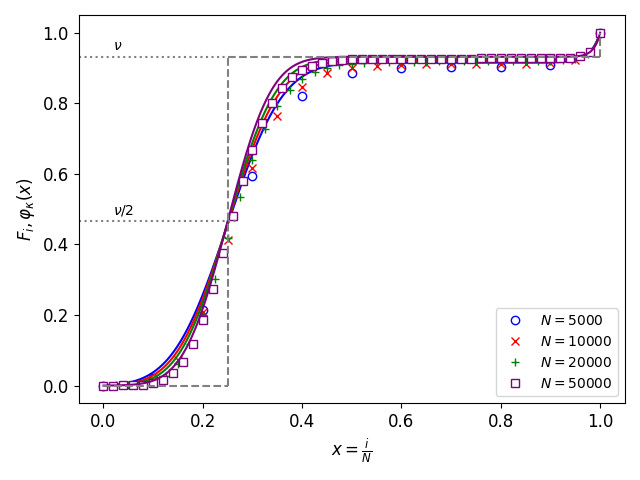}
    \includegraphics[width=0.48\textwidth]{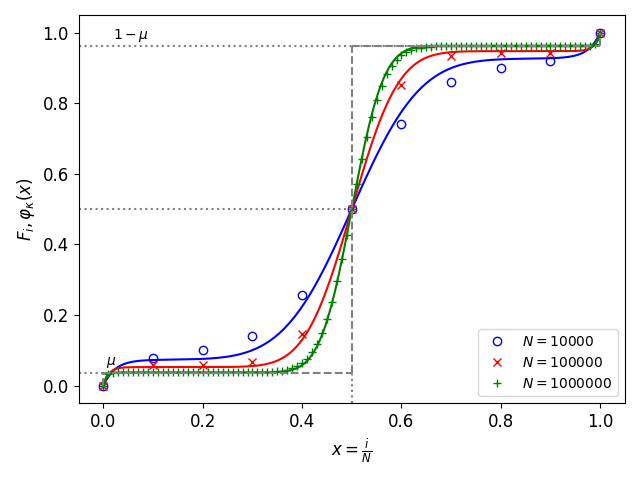}\\
    \includegraphics[width=0.48\textwidth]{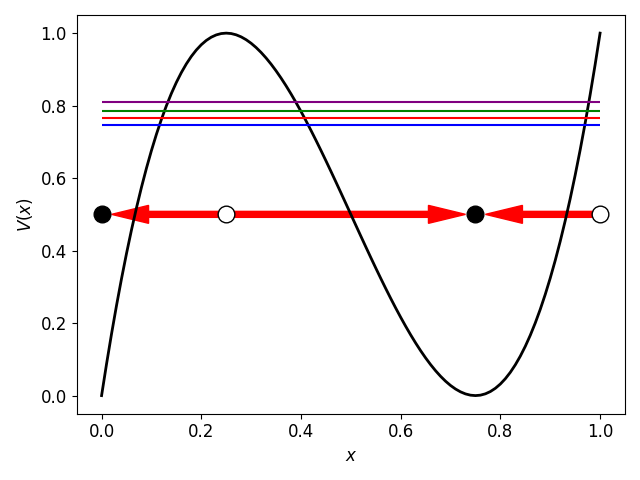}
    \includegraphics[width=0.48\textwidth]{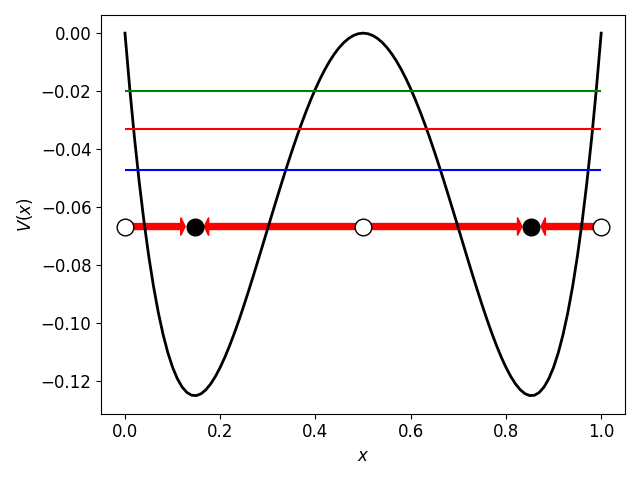}
    \caption{Above: fixation probability for the Moran process (marked) and the asymptotic continuous approximation given by Thm.~\ref{thm:asymptotic}. In both cases, the Moran process and the replicator equation are compatible. Color lines indicate the value of $\langle V\rangle_\kappa$, with consistent colors with the finite population process marked above. Left: $V(x)=9x-24x^2+16x^3$, with two global maxima at $x_-=1/4$ and $x=1$, and two global minima at $x=0$ and $x_{\min}=3/4$. As $V(1/4)=V(1)$, the attractor at $x=3/4$ is an EMS. Furthermore, $f$ as in Def.~\ref{def:compatibility} is such that $f(x_-)=\nu/2$ and $f(x_{\min})=\nu$. Right: $V(x)= -2x+10x^2-16x^3+8x^4$ with two global minima, $x_-$ and $x_+=1-x_-$, and three global maxima, implying the existence of two plateaux. In this case, $f(x_-)=\mu$, $f(x_+)=1-\mu$.}
    \label{fig:plateau}
\end{figure}

\section{The inverse problem}
\label{sec:inverse}

 So far, we have discussed the fixation pattern of the Moran process for given fitness functions and whether it is compatible with the associated replicator equation. In this section, we will discuss the inverse problem, i.e., we will look for potentials $V$ such that the associated Moran process is compatible with a given replicator equation. The formulation of the Replicator Dynamics in terms of the derivative of the fitness potential clearly shows the connection of interior sinks and sources of the replicator flow with the local extrema of the potential. In this vein, let $\mathcal{N}=\{y_0,\dots,y_m\}$ and $\mathcal{M}=\{x_0,\dots,x_M\}$ be the set of distinct stable and unstable equilibria of the replicator, respectively, with $|m-M|\le1$. Since all equilibria are non-degenerated, they are interlaced, i.e., either we have $0=y_0<x_0<y_1<\ldots\le 1$ or $0=x_0<y_0<x_1<\cdots\le 1$. For non-degenerated equilibria, it is always true that we have $\mathcal{N},\mathcal{M}\not=\emptyset$. Indeed, since the boundary points are always equilibria for the replicator, then either $0\in\mathcal{N}$ and $1\in\mathcal{M}$ (or vice-versa) or both belong to a single equilibrium type, $\mathcal{N}$ say, and the interlaced property guarantees that there is an interior equilibrium in $\mathcal{M}$ (or vice-versa).

The inverse problem will be solved by finding $V$ such that its local minima are given by $y_i$ and its local maxima by $x_j$, $i=0,\dots,m$, $j=0,\dots,M$. In addition, one might specify a subset $\widehat{\mathcal{M}}$ of $\mathcal{M}$ specifying the global maxima. No further local minima or maxima should exist.

When $\widehat{\mathcal{M}}=\emptyset$ the inverse problem is easily (and non uniquely) solvable. Let $V'(x)=\pm \prod_{x_i\in\mathcal{N}\cup\mathcal{M}}(x-x_i)$ (the sign is chosen depending on the stability of the first inner equilibrium), then $V(x)=\int_0^x V'(s)\,\rd s$ is a solution. When $\widehat{\mathcal{M}}\not=\emptyset$, then these additional constraints require adding higher-order polynomials to the non-constrained problem, in practice implying increasing the number of players. A general solution does not seem to be available, and we shall investigate a few cases and show that already for some two-player games the solution can be quite complex.

If $\#\mathcal{N}=\#\mathcal{M}=1$, then the replicator dynamics is of dominance type. Its flow is given by \includegraphics[height=0.7\baselineskip]{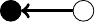} or \includegraphics[height=0.7\baselineskip]{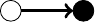}, for dominance by \B or by \A, respectively, where white circles denote unstable equilibria, black circles denote stable equilibria, and the arrows indicates the direction of the flow, and it is not difficult to find a compatible Moran process with a linear potential $V$.

If $\#\mathcal{N}=2$ and $\#\mathcal{M}=1$, then we have a coordination-type game, with the flow of the replicator dynamics given by \includegraphics[height=0.7\baselineskip]{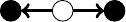}. In this case, a concave quadratic potential with a maximum at the unstable equilibrium, indicating a two-player game, will produce compatible processes.

If $\#\mathcal{N}=1$ and $\#\mathcal{M}=2$, i.e., a coexistence-like game \includegraphics[height=0.7\baselineskip]{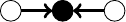}, the situation is more involved, as the constraint $V(0)=V(1)=0$ introduces new restrictions on the choice of the potential. 

To study this class of games, we start with the following lemma:

\begin{lemma}\label{lem:polynomial}
    Let $\mathcal{P}_n$, for $n\in\N$ be the set of strictly positive polynomials in $[0,1]$ with degree at most $n$ and normalized such that $\int_0^1p(x)\rd x=1$. Define 
    \[
    \gamma_n\bydef\inf_{p\in\mathcal{P}_{n}}\int_0^1xp(x)\rd x\ .
    \]
    Then for all $n\in\N\cup\{0\}$, $0<\gamma_{n+1}\le\gamma_n\le \gamma_0=1/2$ and $\lim_n\gamma_n=0$. Furthermore $\sup_{p\in\mathcal{P}_{n}}\int_0^1xp(x)\rd x=1-\gamma_n$.
\end{lemma}

\begin{proof}
The value of $\gamma_0$ is immediate. The fact that the sequence is non-increasing follows from the inclusion $\mathcal{P}_{n}\subset\mathcal{P}_{n+1}$. It is clear that $\gamma_n\ge 0$. We claim that $\gamma_n>0$, otherwise, there is a sequence of polynomials with degree at most $n$ converging to a Dirac delta supported at $x=0$, which is absurd. Finally, since there are polynomial delta-sequences in $[0,1]$ (e.g., $(n+1)(1-x)^n$), we conclude that $\gamma_n\to 0$. The last claim is immediate.
\end{proof}

We use this result to show that given coordinate-like replicator dynamics, with one stable interior point at $\xs\in(0,1)$, there is a polynomial potential, with a sufficiently large degree, such that the associated Moran process has one and only one interior EMS at $\xs$, and no further metastable points.

More precisely, we consider a given replicator equation with two unstable equilibria on the boundaries and one unique stable equilibrium at $x=\xs\in(0,1)$. As we have already seen (Theorem~\ref{thm:asymptotic}), the Moran dynamics associated to a second-order polynomial $V=\gamma(x-\xs)^2-\gamma(\xs)^2$, $\gamma>0$, $\xs\in(0,1)$, is compatible with the replicator if, and only if, $\xs=\frac{1}{2}$. The next result shows $\xs$ becomes closer to either of the boundaries, then the complexity of the interaction (degree of $V$) can increase without bounds. 
\begin{theorem}\label{thm:polynomial}
Let $V$ be a polynomial of degree $d$, with $d\ge 2$, such that the associated Moran process has two unstable equilibria on the boundaries and one unique EMS at $x=\xs$, and no further metastable equilibria. If $\xs\not\in\left[\gamma_n,1-\gamma_n\right]$ then $d> n+2$. 
\end{theorem}

\begin{proof}
    The associated replicator equation is given by $X'=-\lambda X(1-X)(X-X^*)p(X)$, where $\lambda>0$ and $p$ is a polynomial of degree $d-2$, namely $(x-\xs)p(x)=V'(x)$.
    From the assumptions, $V(1)=V(0)=0$, and, therefore
    \[
    \xs=\frac{\int_0^1xp(x)\rd x}{\int_0^1p(x)\rd x}\ .
    \]
    Changing the constant $\lambda$ if necessary, we may assume that $p\in\mathcal{P}_{d-2}$. From Lemma~\ref{lem:polynomial}, $\xs\in(\gamma_{d-2},1-\gamma_{d-2})$, and therefore $\gamma_{d-2}\le\gamma_n$, which concludes the proof.
\end{proof}

As a direct corollary of Lemma~\ref{lem:polynomial}, when the interior equilibrium $\xs$ approaches the boundaries, the minimal degree $d$ of a polynomial that reproduces, for finite populations, the replicator dynamics diverges to infinity. 

It is not difficult to explicitly calculate $\gamma_0=1/2$ and $\gamma_1=1/3$, but to the best of our knowledge, the general problem has not been treated before. It is not difficult to show, that $\gamma_n\le 1/(n+2)$, using, e.g., $p(x)=(n+1)(1-x)^n\in\overline{\mathcal{P}_n}$. However, the inequality is strict at $n=2$, as it can be readily checked by $p(x)=\frac{1200}{331}\left(x^2-\frac{15}{8}x+\frac{88}{100}\right)$, in which case $\xs=\frac{78}{331}<\frac{1}{4}$. In fact, any polynomial of the form $p(x)=\lambda(1-bx+cx^2)$, with $b\in(0,2)$, $\frac{b^2}{4}<c<\frac{5b-4}{6}$, and $\lambda$ a normalization constant, is such that $\int_0^1xp(x)\rd x<\frac{1}{4}$.

When $\xs\in\left[\frac{1}{n},\frac{n-1}{n}\right]$, it is possible to guarantee an upper bound on the minimum degree of the polynomial $V$, as shown in the next result.

\begin{corollary}\label{cor:polynomial}
    Assume $\xs\in\left[\frac{1}{n},1-\frac{1}{n}\right]$. Then, there is a polynomial of degree at most $n$ such that the Moran process has an EMS at $\xs$.
\end{corollary}

\begin{proof}
Let $b_{i,n}(x)\bydef\binom{n}{i}x^i(1-x)^{n-i}\ge 0$ for $x\in[0,1]$, with equality if and only if $x\in\{0,1\}$, be the $i$-th Bernstein polynomials of degree $n\in\mathbb{N}$. The unique local maximum of $b_{i,n}$ is at $x=i/n$. Consider $V(x)=-\sum_{i=1}^{n-1}c_ib_{i,n}(x)$, $c_i\ge 0$ for $i=1,\dots,n-1$, normalized such that $\sum_{i=1}^{n-1}c_i=1$. Then $V(x)\le 0$, $V(0)=V(1)=0$. If $\xs=\frac{i_*}{n}$ for $i_*\in\{1,\dots,n-1\}$, then $c_{i_*}=1$ and $c_i=0$ for $i\ne i_*$. If $\xs\in\left(\frac{i_*}{n},\frac{i_*+1}{n}\right)$, we solve $V'(\xs)=0$, noting that $b'_{i_*,n}(\xs)>0>b'_{i_*+1,n}(\xs)$ to conclude that 
$c_{i_*+1}=-c_{i_*}b'_{i_*,n}/b'_{i_*+1,n}$, with $c_i=0$ otherwise.
\end{proof}

The difference between the simple bounds in Corollary~\ref{cor:polynomial} and the more complicated ones in Theorem~\ref{thm:polynomial} is that any nonnegative combination of Bernstein polynomials is non-negative in the interval $[0,1]$, but there are nonnegative polynomials that, when written in the Bernstein base, have negative coefficients. The simplest example is of second order (i.e., with impact in the calculation of $\gamma_2$), cf.~\cite{Powers_2000}.

If $\xs=1/2$, the second-order polynomial $V(x)=x(x-1)$ will generate a Moran process compatible with $X'=-2X(1-X)(1/2-X)$, i.e., $p(x)=1$. If $\xs\in\left(1/3,2/3\right)$, an explicit calculation shows that $V(x)=x(1-x)(\xs(3\xs-2)+x(1-2\xs))$ and $p(x)=12x\xs-6x-6\xs+4=12(x-1/2)(\xs-1/2)+1\ge 0$. For $x\in[0,1]$ and $\xs\in(1/3,2/3)$ it is possible to prove that $p(x)>0$, as $\left|x-\frac{1}{2}\right|\,\left|\xs-\frac{1}{2}\right|\le\frac{1}{2}\cdot\frac{1}{6}=\frac{1}{12}$, and $\int_0^1xp(x)\rd x=\xs$. In this case $V$ is non-degenerate on the boundaries and therefore, generic.

We conclude that to reproduce any 2-player replicator dynamics in the large, but finite population Moran process, a 2-player game may be enough in the dominance and coordination cases, but for the coexistence case it may be necessary to increase the complexity of underlying dynamics, resourcing possibly to games with an arbitrarily large number of players, depending on the precise position of the interior equilibrium.

\section{Conclusions}
\label{sec:conclusions}
The Moran process and the replicator equation are mathematical models often used in population genetics. The first step in using either model is to describe the reproductive viability, also known as fitness, of all types in every possible population configuration or state. This is often done through game theory, allowing a precise, quantitative description of interactions among individuals. The link between both models assumes non-structured populations, using what is known in physics as a mean-field approximation, see~\cite{Benaim_2003,Sridar_2023,vanKampen} and references therein.

Mean-field approximations significantly simplify the mathematical description of the evolutionary process within a population, ultimately leading to a differential equation, which may be stochastic, partial, or ordinary~\cite{ChalubSouza:TPB2009,Champagnat2006,CzupponTraulsen,TraulsenClaussenHauert_2005}. Other approaches are possible, as, for example, the Moran process can be used in populations structured by contact networks, represented by graphs~\cite{BroomRichtar_2008,Chalub_JDG2016,Lieberman2005,PattniBroom_2015,Shakarian2012}, hierarchically organized in metapopulations~\cite{Constable_2014,ConstableMcKane_2015}, or used for the description of invasion process by rare mutants, the so-called adaptive dynamics~\cite{AllenNowakDieckmann_2013,BragaPereiraWardil_2024,Page2002}.

In all mathematical models of the evolutionary process, the use of game theory is convenient but not necessary. Any fitness function defined on the set of population states is mathematically sufficient for analysis, with the use of potential fitness $V$ providing a useful generalization. Fitness potentials are equivalent to games if and only if $V$ is a polynomial. Furthermore, the conjugate used of fitness potentials and arguments similar to the ones in classical mechanics can further simplify the quantitative analysis of the dynamics of both the Moran process and replicator dynamics, cf.~\cite{ChalubSouza:JTB2018,HansenChalub:JTB2024,PiresNeves_JMB24}. 

Using the framework of game theory whenever convenient, the present work primarily investigates the relationship between the finite population Moran process and the replicator equation. Both models address the evolutionary dynamics of gene frequencies while accounting for differences in reproductive rates. The Moran process incorporates genetic drift due to finite population size, with stochastic fluctuations decreasing when the population increases, whereas the replicator equation describes deterministic dynamics, a framework that implicitly assumes a sufficiently large population~\cite{ChalubCorral_SAM23,ChalubSouza:MCM2008,ChalubSouza:TPB2009} and is expected to model the most probable state of the Markov chain~\cite{Chalubsouza:JMB2014}.

Asymptotic expressions for fixation probability in large populations were constructed in order to relate large, but finite, population Moran processes and the replicator equation. These expressions are directly computed from the fitness potential $V$ and are formulated as convex combinations of two functions related to fixation probabilities in constant fitness and two-player coordination scenarios, cf. Thm~\ref{thm:asymptotic}. Additionally, an explicit and computable expression for the effective population size was derived at Eq.~\eqref{eq:kappaN}, see also Thm.~\ref{thm:continuousfixation}.

In Cor.~\ref{cor:asymptotic}, and, particularly in Thm.~\ref{thm:plateau}, we provided a full characterization of the compatibility between the Moran process and the replicator equation, with the same potential $V$, minus some degenerate cases, in terms only of $V$. We say that these models are compatible when there is a one-to-one correspondence between attractors and repellors. Namely, to what corresponds to an evolutionarily stable state in infinite population dynamics, we introduced the concept of an essentially metastable state, or EMS, in finite population dynamics, generalizing in Def.~\ref{def:compatibility} and Def.~\ref{def:EMS} simultaneously a previous work of the authors where similar concepts were introduced and the well-known $\mathrm{ESS}_N$.

In the following, we studied examples with both zero and positive $V$-gaps, and examined the inverse problem: how to obtain a potential $V$ that corresponds to a given replicator dynamics. While the latter is mathematically straightforward, some unexpected results emerge from the two-player replicator dynamics, indicating that a compatible finite population dynamics may require more complex interactions among individuals than what was assumed in the replicator dynamics. These interactions are represented by payoffs --- and fitnesses --- from $d$-player game theory, with possibly large $d$, even if the replicator equation assumed only 2-player games.

The results presented here are consistent with the so-called Lewontin paradox, the well-known fact that the genetic diversity in a large population is much smaller than would be expected by the neutral theory~\cite{Charlesworth2022-kt,Leffler2012-dt,lewontin1974genetic}.
However, there is no universally accepted explanation for the paradox, with several factors possibly playing a role, e.g., frequent evolutionary bottlenecks~\cite{Charlesworth2009-xx}, genetic hitchhiking~\cite{Smith1974-mw}, background selection~\cite{Charlesworth1995-hm}, among others. See~\cite{Coop2016-uq} for a comprehensive review. Adding to this roll of factors, our results indicate that natural selection with differential reproduction (natural selection) may decisively contribute to the explanation of the paradox --- see~\cite{Buffalo2021-ry} for an alternative point of view.

Furthermore, it is well-known that fixation times are faster in small populations~\cite{Hartle_Clark}. Therefore, we conclude that polymorphisms, here understood as mixed metastable states of the evolutionary dynamics, are rare except, possibly, for populations of intermediate size. Note that there is a clear link between these ideas and the fixation probability of potentials with a positive $V$-gap, in which, for intermediate population size $N$, a fixation plateau appears before being dissolved for large values of $N$.

It should also be noted that it is not difficult to find diversity in natural populations. Indeed, this is a widely studied topic in ecology~\cite{Clark2026-jz,Nolting2016-zv}, an area in which models are typically non-linear and, therefore, other mathematical effects, not considered here, play a relevant role. In the particular case of evolutionary genetics, an interesting example in which polymorphism is maintained is the case of human blood groups, the most well-known being the ABO group~\cite{Logdberg2011-qu}. Although there is no universally accepted explanation for its existence, one possibility is that they have different susceptibility to transmissible disease~\cite{Cserti2007-nq}, i.e., the maintenance of the polymorphism requires an eco-evolutionary interaction.

Our work suggests the existence of an intimate relation between what we term EMS and the classical metastability concept in stochastic processes, cf. \cite{bovier2009metastability} and \cite{miclo2022metastability}, which should be further investigated in forthcoming work. As a simple but illustrative example, notice that for any two-player coexistence game, there is a non-trivial (i.e., supported in the interior of the interval [0,1]) quasistationary distribution that remains so in the deterministic limit $\kappa_N\to0$, cf. \cite[Thm 5.6]{FaureSchreiber} and \cite{meleard2012quasi}. However, the existence of a metastable set $U$ for a given dynamics takes into account the balance between the expected exiting times from $U$ for $x\in U$ and the expected entry time into $U$ for $x\not\in U$ --- see \cite[Def. 4.1 and 4.2]{bovier2009metastability}. Therefore, the existence of quasistationary measures in any Markov chain does not automatically imply that it is dynamically relevant, as it is possible that it became absorbed much faster than the typical hitting time of the quasistationary measure; see also~\cite{AFST_2015_6_24_4_973_0} for the convergence of the distribution conditioned on the non-absorbing set to this quasistationary distribution. In particular,~\cite{AFST_2015_6_24_4_973_0} also shows that for reversible and irreducible processes outside of the absorbing set this speed is controlled by the second spectral gap. Most of the estimates available that imply exponentially long persistence times need a combination of an exponentially small dissipation rate of the quasistationary distribution and a sufficiently fast mixing \cite{AFST_2015_6_24_4_973_0}. This seems not to occur outside EMS points, as exemplified in Fig.~\ref{fig:spectralgap}, where we show the dependence of the spectral gap on the potential $V$ for a general two-player coexistence game.

\begin{figure}
\begin{center}
\includegraphics[width=0.9\textwidth]{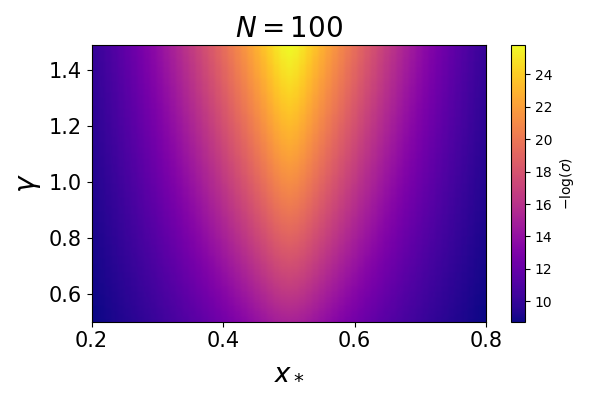}
\caption{Consider a two player game characterized by the matrix $\left(\begin{smallmatrix}1&1+2\gamma\xs\\1+2\gamma(1-\xs)&d\end{smallmatrix}\right)$. The potential is given by Eq.~(12) with $\gamma>0$ and $\xs\in(0,1)$. For $N=100$, we note that the spectral gap $\sigma=1-\lambda$, where $\lambda$ is the subdominant eigenvalue of the Moran transition matrix $\mathbf{M}$, is smaller close to $\xs=1/2$ (horizontal axis), indicating a larger typical existence time for the quasistationary distribution when the potential $V$ is close to symmetric. The vertical axis indicates the value of $\gamma$, showing that for a shallow potential the value of $\sigma$ has a weaker dependence on $\xs$.}
\label{fig:spectralgap}
\end{center}
\end{figure}

We cannot finish this work without comments on the set of examples with zero $V$-gap. As it is shown in Thm.~\ref{thm:plateau}, the condition $\Vgap{V}=0$ is critical for the existence of EMSs. However, this condition is sensitive to small perturbations, i.e., small changes of $V$ make $\Vgap{V}>0$, and Cor.~\ref{cor:asymptotic} applies, and any EMS turns into a (non-essential) metastable state that fades away in the limit $\kappa_N\to0$. More precisely, let $V$ be a generic potential such that $\Vgap{V}=0$, and with exactly two global maxima $x_\pm\in(0,1)$, $x_-<x_+$. Let $\eta$ be an arbitrarily small positive number and consider all generic potentials $\tilde V$ such that $\|\tilde V-V\|_\infty<\eta\ll1$. This set includes examples $V_1$ and $V_2$ such that $V_1(x_-)<V_1(x_+)$ and which $V_2(x_-)>V_2(x_+)$. In the former, the $\varphi_{\kappa}[V_1](x)\to 0$, while in the latter $\varphi_{\kappa}[V_2](x)\to 1$ for $x\in(x_-,x_+)$, indicating a transition between two trivial plateaux, one indicating extinction, the other indicating fixation for functions that are arbitrarily close to $V$. 

To illustrate this last point, consider a symmetric potential, derived from the 4-player game with payoffs $\mathbf{a}^{(4)}=(-5+10, 9+10, -9+10, 5+10)$, and $\mathbf{b}^{(4)}=(10,10,10,10)$. The peculiar way that these payoffs were presented stresses their symmetry. In this case $V(x)=x(1-x)(5-16x(1-x))$, which shows $x\leftrightarrow 1-x$ symmetry, i.e., $V(x)=V(1-x)$. In particular, $F_{i}=1-F_{N-i}$, and, consequently, whenever $N$ is even it is true that $F_{N/2}=1/2$. However, a naive numerical implementation of Eq.~\eqref{eq:Fk} by the authors, in both Python and Julia, shows that the exact fixation probability $F_{N/2}$ approaches zero when $N\to\infty$. On one hand, this shows that new and more robust methods are needed to calculate the fixation probability for large populations, even considering that an explicit and simple formula is available. This happens because in floating-point arithmetic, errors are present in both the estimation of critical points and in the value of $V$ on these points, making $\Vgap{V}$ positive, albeit very small. When using Eq.~\eqref{eq:Fk}, these errors accumulate through the large number of sums and products used to compute $F_{N/2}$. On the other hand, this problem does not show up using the exponential approximation given by Thm.~\ref{thm:continuousfixation}, i.e., $\varphi_{\kappa}(1/2)=1/2$, for all values of $\kappa$, if $V$ is as above. The same is true for the asymptotic approximations described in Thm.~\ref{thm:asymptotic}, as they are obtained from $\varphi_\kappa$. In the end, the asymptotic expression gives a better approximation of the exact value than a naive numerical method using the exact expression, as the former is more robust and the latter involves a large number of computations. As in the humorous dialogue in the introduction of Ref.~\cite{Brujin}, a good approximation shows to be more useful than an exact formula!

\section*{Acknowledgements}

FACCC is funded by Portuguese national funds through the FCT – Fundação para a Ciência e a Tecnologia, I.P. (Portugal), under the scope of the projects UID/297/2025 (https://doi.org/10.54499/UID/00297/2025) and UID/PRR/00297/2025 (https://doi.org/10.54499/UID/PRR/00297/2025) (Center for Mathematics and Applications – NOVA Math). FACCC also acknowledges the support of the project \emph{Mathematical Modelling of Multi-scale Control Systems: Applications to Human Diseases} (CoSysM3)  2022.03091.PTDC (https://doi.org/10.54499/2022.03091.PTDC), supported by national funds (OE), through FCT/MCTES and the COST Action 24122 mSPACE, supported by COST (European Cooperation in Science and Technology, https://www.cost.eu). MOS is supported by CAPES/BR - Finance code 01 and FAPERJ through grant E-26/210.440.2019. All authors acknowledge the support of the CAPES PRINT program at UFF through grant 88881.310210/2018-01.

AI tools were only used to edit and polish the text. The authors checked all the information and are the only ones responsible for any possible incorrectness.

\section*{Authorship and Contributorship Statement}

FACCC and MOS contributed equally to the conception, development, and writing of this manuscript and share equal responsibility for its content.

%\bibliographystyle{abbrv}
%\bibliography{References}

\appendix

\section{The replicator equation as the large population limit of the Moran process}
\label{ap:formal}

In this appendix, we show that the infinite population limit of the Moran process, for a suitably defined choice of scaling, is the Kimura equation. This proof is somehow implicit in two previous results from the authors: in~\cite{ChalubSouza:TPB2009} we used a more \emph{ad hoc} approach, identifying the Moran process away of the absorbing states as a numerical discretisation of the Kimura equation over $(0,1)$ and connecting the Dirac's with the help of theory in \cite{ChalubSouza:CMS2009}. A more systematic approach appeared in~\cite{Chalubsouza:JMB2014}, using a weak-like formulation for a Markov process, but for the case of the Wright-Fisher process. Here we present a proof using the weak formulation approach used in \cite{Chalubsouza:JMB2014}. We also point out that the results in \cite{Chalubsouza:JMB2014} show that, for short times, the initial dynamics of the Kimura equation (and hence of the Wright-Fisher and Moran processes) for an initial probability distribution of states is concentrated around a single, well-defined state, which corresponds to the characteristics of the replicator equation.

Consider the fitness functions given by 
\[
\Psi^{(\X)}(xN)=1+\left(\Delta t\right)^{\nu}\psi^{(\X)}(x)\ ,
\]
where $\X=\A,\B$, and $x\in\{0,1,\dots,N\}$. Note that, following our previous works, and to motivate the use of scaling, we kept an arbitrary time-step $\Delta t$ for now. Using the definition of the fitness potential given by Eq.~\eqref{eq:def_V}, we find that
\begin{equation*}
 \frac{x\Psi^{(\A)}(xN)}{x\Psi^{(\A)}(xN)+(1-x)\Psi^{(\B)}(xN)}=x\left[1-\left(\Delta t\right)^\nu(1-x)V'(x)+\smallo{(\Delta t)^\nu}\right]\ .
\end{equation*}
For the transition matrix, we find, up to $\smallo{(\Delta t)^\nu}$, using $x=i/N$,
\begin{align*}
&M_{i,i-1}=x(1-x)\left(1+(\Delta t)^\nu xV'(x)\right)\ ,\\
&M_{i,i+1}=x(1-x)\left(1-(\Delta t)^\nu (1-x)V'(x)\right)\ .
\end{align*}

Define $\mathbf{p}(t)=(p_0(t),\dots,p_N(t))$, where $p_i(t)$ is the probabilitly to find the system at state $i$ at time $t$. The Moran process evolves according to the Master equation $\mathbf{p}(t+\Delta t)=\mathbf{Mp}(t)$. We introduce the notation $\mathbb{N}_N=\{0,N^{-1},2N^{-1},\dots,N\}$, and $\mathbb{T}_\ell=\{\ell\Delta t, \ell\in\mathbb{N}\cup\{0\}\}$.

We introduce the notation $z\bydef\frac{1}{N}$.
For $i=xN$ and after multiplication by a test function $q\in C^\infty([0,1]\times(0,\infty))$ and sum over all possible states, we find
 \begin{align*}
  \sum_{x\in N^{-1}\N_{N+1},t\in\T_\ell}M_{i-1,i}p(x-z,t)q(x,t)
&  =\sum_{x\in N^{-1}\N_{N+1},t\in\T_\ell}M_{i,i+1}p(x,t)q(x+z,t)\ ,\\
  \sum_{x\in N^{-1}\N_{N+1},t\in\T_\ell}M_{i+1,i}p(x+z,t)q(x,t)
&  =\sum_{x\in N^{-1}\N_{N+1},t\in\T_\ell}M_{i,i-1}p(x,t)q(x-z,t)\ ,
 \end{align*}
We assume $M_{ij}=0$ if $i,j<0$ or $i,j>N$. 

Adding over all possible times:, we find
\begin{align}
\label{ad:eq:prod_int_fraco}
& \sum_{x\in N^{-1}\N_{N+1},t\in\T_\ell}p(x,t+\Delta t)q(x,t)\\
\nonumber
&\quad=\sum_{x\in N^{-1}\N_{N+1},t\in\T_\ell}\left[ M_{i-1,i}p(x-z,t)+M_{i,i}p(x,t)+M_{i+1,i}p(x+z,t), q(x,t)\right]\\
\nonumber
&\quad=\sum_{x\in N^{-1}\N_{N+1},t\in\T_\ell}p(x,t)\left\{M_{i,i+1}q(x+z,t)+M_{i,i}q(x,t)+M_{i,i-1}q(x-z,t)\right\}\\
\nonumber
&\quad=\sum_{x\in N^{-1}\N_{N+1},t\in\T_\ell}p(x,t)\Bigl\{\underbrace{\left(M_{i,i+1}+M_{i,i}+M_{i,i-1}\right)}_{=1}q(x,t)\\
\nonumber
&\qquad+z\underbrace{\left(M_{i,i+1}-M_{i,i-1}\right)}_{=-x(1-x)V'(x)(\Delta t)^\nu+\smallo{(\Delta t)^\nu}}\partial_xq(x,t)+\frac{z^2}{2}\underbrace{\left(M_{i,i+1}+M_{i,i-1}\right)}_{=x(1-x)+\smallo{1}}\partial_x^2q(x,t)\Bigr\}+\smallo{z^3}\ .
\end{align}

We gather all results in Equation~\eqref{ad:eq:prod_int_fraco} and define the operation. 
\[
\mathcal{T}_{\Delta t}p(x,t)\bydef p(x,t+\Delta t)\ .
\]
For any test function $q$, it is true that
\begin{equation*}
\langle\mathcal{T}_{\Delta t} p,q\rangle\\
 =\left\langle p,q-z\left(\Delta t\right)^\nu x(1-x)V'(x)\partial_xq+\frac{z^2}{2}x(1-x)\partial_x^2q\right\rangle\\
 +\smallo{z^2,z\left(\Delta t\right)^\nu}\ .
\end{equation*}
We used the notation $\smallo{a,b}=\smallo{a}+\smallo{b}$.

Using 
$\langle\mathcal{T}_{\Delta t} p,q\rangle=\langle p,\mathcal{T}_{-\Delta t}q\rangle$\,
we conclude 
\begin{equation*}
 \left\langle p,\frac{\mathcal{T}_{-\Delta t}q-q}{\Delta t}\right\rangle=\left\langle p,-z\left(\Delta t\right)^{\nu-1} x(1-x)V'(x)\partial_xq+z^2\left(\Delta t\right)^{-1}x(1-x)\partial_x^2q\right\rangle+\smallo{z^2\left(\Delta t\right)^{-1},z\left(\Delta t\right)^{\nu-1}}\ .
\end{equation*}

Using the scaling relationship $N^{-1}=z=\frac{\varepsilon}{2}\left(\Delta t\right)^{\mu}$, we conclude, up to order
$\smallo{\left(\Delta t\right)^{2\mu-1},\left(\Delta t\right)^{\mu+\nu-1}}$, it is not difficult to verify that
\begin{equation*}
 \left\langle p,\frac{\mathcal{T}_{-\Delta t}q-q}{\Delta t}\right\rangle=\Biggl\langle p,-\frac{\varepsilon}{2}\left(\Delta t\right)^{\mu+\nu-1}x(1-x)V'(x)\partial_xq+\frac{\varepsilon^2}{4}\left(\Delta t\right)^{2\mu-1}x(1-x)\partial_x^2q\Biggr\rangle\ .
\end{equation*}

We finally fix values for $\mu$, $\nu$ and consider the limit $N\to\infty$, $\Delta t\to0$, and changing the time scale $t\to \frac{\varepsilon t}{2}$, we find different limit equations, for $t>0$:
\begin{enumerate}
 \item If $\mu+\nu=1$ e $\mu>\frac{1}{2}$:
 \begin{equation}\label{ad:eq:replicador_pde_weak}
-\langle p,\partial_t q\rangle=-\left\langle p,x(1-x)V'(x)\partial_x q\right\rangle\ .
 \end{equation}
 \item If $\mu=\frac{1}{2}$ e $\nu>\frac{1}{2}$:
 \begin{equation}\label{ad:eq:pure_diffusion_weak}
-\langle p,\partial_t q\rangle=\frac{\varepsilon}{2}\left\langle p,x(1-x)\partial_x^2q\right\rangle\ .
 \end{equation}
\item If $\mu=\nu=\frac{1}{2}$:
\begin{equation}\label{ad:eq:Kimura_weak}
-\langle p,\partial_t q\rangle=\left\langle p,-x(1-x)V'(x)\partial_x q+\frac{\varepsilon}{2}x(1-x)\partial_x^2q\right\rangle\ .
\end{equation}
\end{enumerate}
For all other possible values of $\mu,\nu$ the resulting equation is ill-defined or trivial.

 Eq.~\eqref{ad:eq:replicador_pde_weak} corresponds to the one studied in the present work, cf. Rmk.~\ref{rmk:scalings}, and is a weak formulation of the Replicator Equation~\eqref{eq:replicator}.

Eq.~\eqref{ad:eq:Kimura_weak} is the generalization for arbitrary fitness functions of the weak formulation of the Kimura Equation, introduced in~\cite{Kimura_1958,Kimura62}. This equation was derived as the large population limit of the Moran process in~\cite{ChalubSouza:TPB2009}, and for the Wright-Fisher process in~\cite{Chalubsouza:JMB2014}.  Its rigorous mathematical study was performed in~\cite{ChalubSouza:CMS2009,ChugonovaTaranets_2023,DanilkinaSouzaChalub_2018} and reformulated as a gradient flow in~\cite{ChalubMonsaingeon:AAP2021}.

\section{Derivation of the asymptotic expression in the two-player game coordination case, Eq.~\eqref{eq:varphiC}.}
\label{ap:varphiC}

We define 
\begin{displaymath}
\II_\kappa[V](x) = \int_0^x\re^{\frac{2}{\kappa}V(s)\,\rd s}
\end{displaymath}
such that $\varphi_\kappa[V](x)=\II_\kappa[V](x)/\II_\kappa[V](1)$.

We assume that $V$ is $C^3([0,1])$ and that it has a unique maximum at $\xs$ $\in(0,1)$. We will use a slightly modified version of Laplace's method: Write

\[
V(x) = V(\xs)-|V''(\xs)|\frac{(x-\xs)^2}{2} + R_3(x),
\]
and let $\alpha=\min_{x \in [0,1]} V'''$ and $\beta=\max_{x\in[0,1]} V'''$. Let 
\begin{displaymath}
\sigma_\kappa^{-1} \bydef \sqrt{\frac{2|V''(\xs)|}{\kappa}},\quad \sigma_\kappa z\bydef x-\xs
\end{displaymath}
Then, from the Lagrange form of the remainder for Taylor series, we have
\begin{displaymath}
   \kappa \alpha z^3\leq \sigma_\kappa^{-1}R_3(\xs+\sigma_\kappa z)\leq \kappa \beta z^3.
\end{displaymath}
Therefore, we have
\begin{align*}
\II_\kappa[V](x) &= \sigma_\kappa\re^{\frac{2}{\kappa}V(\xs)}\int_{-\sigma_\kappa^{-1}\xs}^{\sigma_\kappa^{-1}(x-\xs)}\re^{-\frac{z^2}{2}}\re^{\sigma_\kappa^{-1}R_3(\xs+\sigma_\kappa z)}\,\rd z,\\
&= \sigma_\kappa\re^{\frac{2}{\kappa}V(\xs)}\left\{\int_{-\sigma_\kappa^{-1}\xs}^{\sigma_\kappa^{-1}(x-\xs)}\re^{-\frac{z^2}{2}}\,\rd z + \order{\kappa^{2}}\right\},\\
&=\sigma_\kappa\sqrt{2\pi}\re^{\frac{2}{\kappa}V(\xs)}\left\{\mathcal{N}\left(\sigma_\kappa^{-1}(x-\xs)\right)- \mathcal{N}\left(-\sigma_\kappa^{-1}\xs\right)+ \order{\kappa^{2}}\right\}.
\end{align*}
The second equality follows by expanding $\exp(\kappa \beta z^3)$ and performing repeated integration by parts.

Hence,
\begin{displaymath}
    \varphi_\kappa[V](x) = \varphi_\kappa^{(C)}[\xs,V''(\xs)](x)+\order{\kappa^2}.
\end{displaymath}

\section{Proof of Theorem~\ref{thm:asymptotic}}
\label{ap:ThmAsymptotic}

In what follows, we assume that, if $x$ is a local maximum of $V$, then $V'(x)=\order{1}$, provided that $x$ is an endpoint; otherwise, $V''(x)=\order{1}$.

We begin by proving an auxiliary lemma:

\begin{lemma}
\label{lem:aux:thm2}
    Let  $0\leq a<b\leq1$ and 
\begin{displaymath}
    \II_\kappa[V;a,b](x)\bydef \left\{
    \begin{array}{cc}
       0  & 0\leq x \leq a \\
       \int_a^x\re^{\frac{2}{\kappa}V(s)}\,\rd s  & a < x < b\\
       \int_a^b\re^{\frac{2}{\kappa}V(s)}\,\rd s  & b \leq x\leq 1.
    \end{array}
    \right.
\end{displaymath}
Then, if $V'<0$ on $[0,b]$ we have
\begin{displaymath}
    \II_\kappa[V;0,b](x)= \frac{\kappa}{2|V'(0)|}\left(1-\re^{-2|V'(0)|x/\kappa}\right) + \order{\kappa^2}
\end{displaymath}
If $V'>0$ on $[a,1]$, we have
\begin{displaymath}
    \II_\kappa[V;a,1](x)= \frac{\kappa\re^{\frac{2}{\kappa}V(1)}}{2V'(1)}\left\{\left(\re^{-2V'(1)(1-x)/\kappa}-\re^{-2V'(1)(1-a)/\kappa}\right) + \order{\kappa}\right\}.
\end{displaymath}
Finally, if $V'>0$ on $(a,\xs)$ and $V'<0$ on $(\xs,b)$, we have 
\begin{displaymath}
    \II_\kappa[V;a,b](x)=\sigma_\kappa\sqrt{2\pi}\re^{\frac{2}{\kappa}V(\xs)}\left\{\left[\mathcal{N}\left(\sigma_\kappa^{-1}(x-\xs)\right)- \mathcal{N}\left(\sigma_\kappa^{-1}(a-\xs)\right)\right]+ \order{\kappa^{2}}\right\}.
\end{displaymath}
\end{lemma}

\begin{proof}
        In the first case, the unique maximum of $V$  over $[0,b]$ is at $x=0$ and
\begin{displaymath}
V(x) = V'(0)x + R_2(x),
\end{displaymath}
Let $\kappa z = x$ and notice that we have
\begin{displaymath}
  \kappa  \alpha z^2\leq \kappa^{-1} R_2(\kappa z)\leq \kappa \beta z^2.
\end{displaymath}
Therefore:
\begin{align*}
    \II_\kappa[V](x)&= \kappa\int_0^{x/\kappa}\re^{-2|V'(0)|z}\re^{\kappa^{-1}R_2(\kappa z)}\,\rd z,\\
    &=\kappa\int_0^{x/\kappa}\re^{-2|V'(0)|z}\,\rd z + \order{\kappa^2},\\
    &= \frac{\kappa}{2|V'(0)|}\left(1-\re^{-2|V'(0)|x/\kappa}\right) + \order{\kappa^2},
\end{align*}
from which the result follows.

The proof for the second case is similar, and for the third case it follows, \emph{mutatus mutandis}, from the proof in Appendix~\ref{ap:varphiC}.
\end{proof}
Let
\begin{displaymath}
    \varphi_\kappa[V,a,b](x):= \frac{\II_\kappa[V,a,b](x)}{\II_\kappa[V,a,b](b)}.
\end{displaymath}
Then, we have that  $\varphi_\kappa[V,a,b]$ is non-decreasing over $[0,1]$ and strictly monotonic over $[a,b]$, with  $\varphi_\kappa[V,a,b](a)=0$ and  $\varphi_\kappa[V,a,b](b)=1$ 

\begin{lemma}
\label{thm:aux:thm2}
Let $V:[0,1]\to\R$, such that $V$ is at least $C^3([0,1])$ and that $V$ has a finite number of local maxima in $[0,1]$. Let $\mathcal{M}=\{x_0,x_1,\ldots,x_M\}$ be the set of local maxima of $V$. Moreover, assume further that
\begin{enumerate}
    \item if $x_0\not=0$, then $|x_0|\gg \sqrt{\kappa}$;
    \item if $x_m\not=1$, then $|1-x_m|\gg \sqrt{\kappa}$;
    \item $|x_k-x_{k+1}|\gg \sqrt{\kappa}$, $k=0,\ldots,m-1$.
\end{enumerate}
Furthermore, if $x_0\not=0$ and $x_m\not=1$, let $0=y_0\leq x_0<y_1<x_1<\ldots<y_m<x_m\leq y_{m+1}=1$, such that $|y_k-x_{k-1}|,|y_k-x_k|\gg\sqrt{\kappa}$, $k=1,\ldots,m$, and such that $V’$ has a unique zero in $(y_j,y_{j+1})$. If $x_0=0$, then start with $y_1$, and if $x_M=1$ then end at $y_M$. Then

\begin{equation}
\label{eq:ap:thm2}
    \varphi_\kappa[V](x) = \sum_{i=0}^M c_i \left(\varphi_\kappa[V,y_i,y_{i+1}](x)\right)^+ \land 1 + \order{\kappa}
\end{equation}
where
\begin{displaymath}
    x^+=\max(x,0),\quad a\land b= \min(a,b)\quad\text{and}\quad c_i= \frac{I_\kappa[V,y_i,y_{i+1}](y_{i+1})}{\sum_{j=0}^M I_\kappa[V,y_j,y_{j+1}](y_{j+1}) }.
\end{displaymath}
\end{lemma}

\begin{proof}
The proof follows by noticing that 
\begin{align*}
    \left.\varphi_\kappa[V](x)\right|_{[y_i,y_{i+1}]}&= \varphi_\kappa[V](y_i)+ \frac{I_\kappa[V,y_i,y_{i+1}](x)}{\sum_{j=0}^MI_\kappa[V,y_j,y_{j+1}](y_{j+1})}\\
    &= \sum_{j=0}^{i-1}c_j + c_i \varphi_\kappa[V,y_i,y_{i+1}] + \order{\kappa}\\
    &= \sum_{j=0}^i c_i \left(\varphi_\kappa[V,y_i,y_{i+1}](x)\right)^+ \land 1 + \order{\kappa}.
\end{align*}
Above, the  first equality comes from $I_\kappa[V,0,1]=\sum_{j=0}^mI_\kappa[V,y_j,y_{j+1}](y_{j+1})$, the second equality follows from Lemma~\ref{lem:aux:thm2} and the third equality comes from the monotonicity of $\varphi_\kappa[V,a,b]$ together with its attained values at $x=a,b$. 
\end{proof}

\begin{proof}[Proof of Theorem~\ref{thm:asymptotic}]
    The proof follows from Lemma~\ref{thm:aux:thm2}, given that $M=\order{1}$. In this case, we can replace $\left(\varphi_\kappa[V,y_i,y_{i+1}](x)\right)^+ \land 1 $ by the appropriate $\varphi_\kappa^{(\X})$, where $\X=\A,\B$ or $C$,  with an error $M\order{\kappa} = \order{\kappa}$ and also replace $c_i$ by $d_i$ with a similar estimate. 
\end{proof}

\begin{remark}
    While Thm.~\ref{thm:asymptotic} yields a more natural representation of the fixation probability, it does require tighter estimates of the integral approximations.  Also note that, under an additional assumption of asymptotic separability of all local extrema of $V$ over $[0,1]$, the $y_i$ can be chosen to be the local minima of $V$. 
\end{remark}
\end{document}